\documentclass{amsart}

\usepackage{amsmath,amsthm,amsbsy,amscd,amssymb,graphicx, epsfig,color,xy}

\usepackage[latin1]{inputenc}
\usepackage{array}
\usepackage{mathrsfs}
\usepackage{float}

\usepackage{hhline}
\usepackage{multirow}
\usepackage{multicol}

\xyoption{all}
\theoremstyle{plain}

\newtheorem{theorem}{Theorem}
\newtheorem{lemma}[theorem]{Lemma}
\newtheorem{proposition}[theorem]{Proposition}

\theoremstyle{definition}

\newtheorem{definition}[theorem]{Definition}
\newtheorem{example}[theorem]{Example}
\newtheorem{algorithm}[theorem]{Algorithm}
\newtheorem{remark}[theorem]{Remark}
\newtheorem{notation}[theorem]{Notation}

\usepackage{tikz}
\usetikzlibrary{arrows,backgrounds,positioning,matrix,fit,decorations.pathmorphing,calc,
decorations.markings,decorations.pathreplacing,shapes}

\newcommand{\R}{\mathbb{R}}
\newcommand{\Q}{\mathbb{Q}}
\newcommand{\Z}{\mathbb{Z}}

\newcommand{\Lap}{\mathcal{L}}

\newcommand{\Sp}{\mathscr{S}}
\newcommand{\Edg}{\mathcal{E}}
\newcommand{\Ve}{\mathcal{V}}
\newcommand{\Ka}{\boldsymbol{\kappa}}
\newcommand{\Katau}{\boldsymbol{\tau}}
\newcommand{\uri}{\rightarrow_\circ}

\newcommand{\cte}{C}
\newcommand{\cond}{\mathcal{C}}
\newcommand{\bu}{\mathbf{u}}
\newcommand{\bx}{x}
\newcommand{\stoich}{S}
\newcommand{\stoichM}{M}
\newcommand{\stoichMp}{M^\bot}
\newcommand{\bino}{T}
\newcommand{\binoM}{B}
\newcommand{\binoMp}{B^\bot}

\newcommand{\intal}{\mathrm{Int}}
\newcommand{\GE}{G_E}
\newcommand{\ww}{\mathbf{w}}
\newcommand{\vv}{\mathbf{v}}

\DeclareMathOperator{\sign}{sign}
\DeclareMathOperator{\diag}{diag}
\DeclareMathOperator{\rank}{rank}

\usepackage{hyperref}
{\thispagestyle{empty}} 

\title{The structure of MESSI biological systems}

\author[M. P\'erez Mill\'an and A. Dickenstein]{Mercedes P\'erez Mill\'an*,** and Alicia Dickenstein*,**}
\address{* Dto.\ de Matem\'atica, FCEN, Universidad de Buenos Aires,
    Ciudad Universitaria, Pab.\ I, C1428EGA Buenos Aires, Argentina.}
\address{** IMAS  (UBA-CONICET), Ciudad Universitaria,
Pab.\ I, C1428EGA Buenos Aires, Argentina.}
\email{alidick@dm.uba.ar, mpmillan@dm.uba.ar}
\date{}

\thanks{This work was partially supported by UBACYT 20020100100242, CONICET PIP 11220150100473 and 11220150100483, and ANPCyT PICT 2013-1110, Argentina.} 
\begin{document}

\begin{abstract}
We introduce a general framework for biological systems, called MESSI systems, 
that describe Modifications of type Enzyme-Substrate or Swap with Intermediates, 
and we prove general results based on the network structure.
Many posttrans\-la\-tion\-al modification networks are MESSI systems. 
Examples are the motifs in [Feliu and Wiuf (2012a)], 
sequential distributive and processive multisite phosphorylation networks, 
most of the examples in~[Angeli et al.~(2007)],
phosphorylation cascades, 
two component systems as in~[Kothamachu et al.~(2015)], 
the bacterial EnvZ/OmpR network in~[Shinar and Feinberg (2010)], 
and all linear networks. We show that,
under mass-action kinetics, MESSI systems are conservative. We simplify the study
of steady states of these systems by explicit elimination of intermediate complexes
and we give conditions to ensure an explicit rational parametrization of the variety
of steady states (inspired by [Feliu and Wiuf (2013a, 2013b), Thomson and Gunawardena (2009)]).  
We define an important subclass of MESSI systems with
toric steady states [P\'erez Mill\'an et al.~(2012)]
and we give for MESSI systems with toric steady states an
easy algorithm to determine the capacity for multistationarity. In this case,  
the algorithm provides rate constants for which multistationarity takes place, 
based on the theory of oriented matroids.
\end{abstract}

%

\maketitle


\section{Introduction}

\medskip

Many processes within cells involve some kind of posttrans\-la\-tion\-al modification of proteins.
We introduce a general framework for biological systems that describe Modifications of type
Enzyme-Substrate or Swap with Intermediates, which we call MESSI systems, and which allows us
to prove general results on their dynamics from the structure of the network, under mass-action kinetics.
This subclass of mechanisms has attracted considerable theoretical attention due to its abundance in nature
and the special characteristics in the topologies of the networks.

The basic idea in the definition of MESSI systems (see Definitions~\ref{def:messi} and~\ref{def:ms}) is that the 
mathematical modeling reflects the different chemical behaviors. The chemical species can be grouped into
different subsets according to the way they participate in the reactions, very much akin to the intuitive
partition of the species according to their function.
We show that MESSI systems are conservative (and thus all trajectories are defined for any positive time), 
and we study the important questions of persistence and multistationarity. 
Informally, persistence means that no species which is present 
can tend to be eliminated in the course of the reaction \cite{ADLS07}.
Multistationarity (see Definition~\ref{def:multi}) is also a crucial property, 
since its occurrence can be thought of as a mechanism for switching between 
different response states in cell signaling systems and enables multiple outcomes
for cellular decision making, with the same stoichiometric content.

Examples of MESSI systems of major biological importance are phosporylation cascades,
such as the mitogen-activated protein kinases (MAPKs)
cascades \cite{CDVS16,sig-016,kholo00}. MAPKs are serine/threonine kinases
that play an essential role in signal transduction by modulating gene transcription in
the nucleus in response to changes in the cellular environment and participate in a
number of disease states including chronic inflammation and cancer
\cite{davis,kyriakis,pearson,schaeffer,zarubin} as they control key cellular functions
\cite{hornberg,pearson,schoeberl,tvg07,widmann}.
Also the {\em multisite phosphorylation system} is a MESSI system.
This network describes the phosphorylation of a protein in multiple sites by a
kinase/phosphatase pair in a sequential and distributive mechanism 
\cite{cyc-007,NFAT-002,NFAT-001,sig-016,NFAT-003,sig-051}.
In prokaryotic cells, an example of a MESSI system can be found in \cite{sf10},
representing the \emph{Escherichia coli} EnvZ-OmpR system which
consists of the sensor kinase EnvZ, and the response-regulator OmpR 
(see also \cite{HsSi00,IgNi89,PrSi95,StRo00,ZhQi00}).
This signaling system is a prototypical two-component signaling
system \cite{PrSi95,StRo00}. All linear systems are also MESSI.

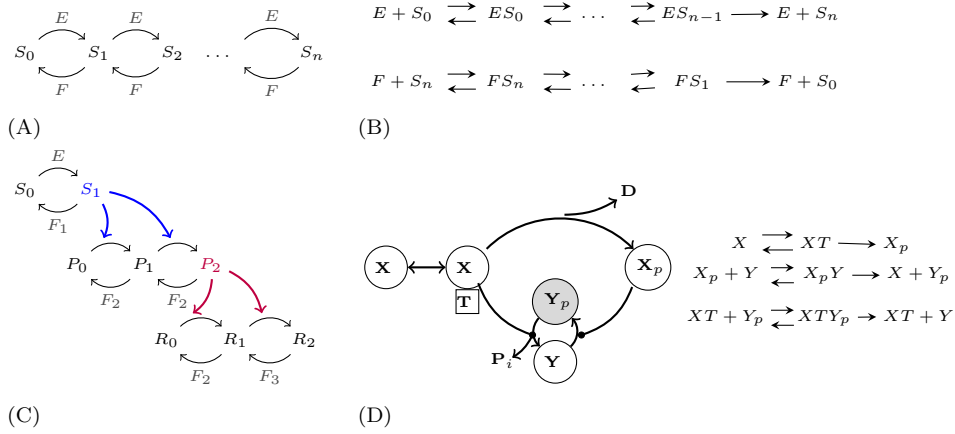
\begin{figure}[t!]  
\begin{tabular}{ll}
{\tiny 
\begin{tikzpicture}[node distance=0.5cm]
  \node[] at (1.1,-1.5) (dummy2) {};
 \node[left=of dummy2] (p0) {$S_0$};
 \node[right=of p0] (p1) {$S_1$}
   edge[->, bend left=45] node[below] {\textcolor{black!70}{\tiny{$F$}}} (p0)
   edge[<-, bend right=45] node[above] {\textcolor{black!70}{\tiny{$E$}}} (p0);
 \node[right=of p1] (p2) {$S_2$}
   edge[->, bend left=45] node[below] {\textcolor{black!70}{\tiny{$F$}}} (p1)
   edge[<-, bend right=45] node[above] {\textcolor{black!70}{\tiny{$E$}}} (p1);
 \node[right=0.1cm of p2] (p4) {$\dots\phantom{S_2}$};
 \node[right=0.4cm of p4] (p5) {$S_n$}
   edge[->, bend left=45] node[below] {\textcolor{black!70}{\tiny{$F$}}} (p4)
   edge[<-, bend right=45] node[above] {\textcolor{black!70}{\tiny{$E$}}} (p4);
\end{tikzpicture}} &
{\tiny
 \begin{tikzpicture}
  \matrix (m) [matrix of math nodes, row sep=2em,ampersand replacement=\&,
    column sep=2em]{
    E+S_0 \& ES_0\& \vphantom{ES_0}\dots\vphantom{ES_0} \& ES_{n-1} \& E+S_n \\
    F+S_n \& FS_n \& \vphantom{ES_0}\dots\vphantom{ES_0} \& FS_1 \& F+S_0 \\};
  \path[-stealth] (m-1-4) edge node[above] {} (m-1-5)
    ($(m-1-1)+(0.6,0.1)$) edge node[above] {}  ($(m-1-2)+(-0.4,0.1)$)
    ($(m-1-2)+(-0.4,-0.1)$) edge node[below] {}  ($(m-1-1)+(0.6,-0.1)$)
    ($(m-1-2)+(0.5,0.1)$) edge node[above] {}  ($(m-1-3)+(-0.3,0.1)$)
    ($(m-1-3)+(-0.3,-0.1)$) edge node[below] {}  ($(m-1-2)+(0.5,-0.1)$)
    ($(m-1-3)+(0.5,0.1)$) edge node[above] {}  ($(m-1-4)+(-0.5,0.12)$)
    ($(m-1-4)+(-0.5,-0.08)$) edge node[below] {}  ($(m-1-3)+(0.5,-0.1)$);
  \path[-stealth] ($(m-2-4)+(0.45,0)$) edge node[above] {} (m-2-5)
    ($(m-2-1)+(0.6,0.1)$) edge node[above] {}  ($(m-2-2)+(-0.4,0.1)$)
    ($(m-2-2)+(-0.4,-0.1)$) edge node[below] {}  ($(m-2-1)+(0.6,-0.1)$)
    ($(m-2-2)+(0.5,0.1)$) edge node[above] {}  ($(m-2-3)+(-0.3,0.1)$)
    ($(m-2-3)+(-0.3,-0.1)$) edge node[below] {}  ($(m-2-2)+(0.5,-0.1)$)
    ($(m-2-3)+(0.5,0.1)$) edge node[above] {}  ($(m-2-4)+(-0.5,0.12)$)
    ($(m-2-4)+(-0.5,-0.08)$) edge node[below] {}  ($(m-2-3)+(0.5,-0.1)$);
\end{tikzpicture}}\\
{\footnotesize (A)}& {\footnotesize (B)}\\
{\tiny
 \begin{tikzpicture}[node distance=0.4cm]
  \node[] at (2.2,-2) (dummy3) {};
  \node[left=of dummy3] (r0) {$R_0$};
  \node[right=of r0] (r1) {$R_1$}
    edge[->, bend left=45] node[below] {\textcolor{black!70}{\tiny{$F_2$}}} (r0)
    edge[<-, bend right=45] node[above] {} (r0);
  \node[right=of r1] (r2) {$R_2$}
    edge[->, bend left=45] node[below] {\textcolor{black!70}{\tiny{$F_3$}}} (r1)
    edge[<-, bend right=45] node[above] {} (r1);
  \node[] at (1,-1) (dummy2) {};
  \node[left=of dummy2] (p0) {$P_0$};
  \node[right=of p0] (p1) {$P_1$}
    edge[->, bend left=45] node[below] {\textcolor{black!70}{\tiny{$F_2$}}} (p0)
    edge[<-, bend right=45] node[above] {} (p0);
  \node[right=of p1] (p2) {\textcolor{purple}{$P_2$}}
    edge[->, bend left=45] node[below] {\textcolor{black!70}{\tiny{$F_2$}}} (p1)
    edge[<-, bend right=45] node[above] {} (p1)
    edge[->,thick,color=purple, bend left=25] node[above] {} ($(r0.north)+(10pt,5pt)$)
    edge[->,thick,color=purple, bend left=25] node[above] {} ($(r1.north)+(10pt,5pt)$);
  \node[] at (0.3,0) (dummy) {};
  \node[left=of dummy] (s0) {$S_0$};
  \node[right=of s0] (s1) {\textcolor{blue}{$S_1$}}
    edge[->, bend left=45] node[below] {\textcolor{black!70}{\tiny{$F_1$}}} (s0)
    edge[<-, bend right=45] node[above] {\textcolor{black!70}{\tiny{$E$}}} (s0)
    edge[->,thick,color=blue, bend left=25] node[above] {} ($(p0.north)+(10pt,5pt)$)
    edge[->,thick,color=blue, bend left=25] node[above] {} ($(p1.north)+(10pt,5pt)$);
  \end{tikzpicture}} &
{\tiny
 \begin{tikzpicture}[node distance=0.4cm]
  \node[]at(0,0)(a){};
  \node[right=3pt of a] (red) {
    \begin{tikzpicture}
      \node[]at(0,0)(a){};
      \node[above=-8pt of a] (f1) {
      \begin{tikzpicture}
    \matrix (m) [matrix of math nodes, row sep=1em,ampersand replacement=\&,
    column sep=2em]{
    X \& XT\& X_p \\};
    \path[-stealth]
      (m-1-2) edge (m-1-3)
      ($(m-1-1)+(0.3,0.1)$) edge ($(m-1-2)+(-0.3,0.1)$)
      ($(m-1-2)+(-0.3,-0.1)$) edge ($(m-1-1)+(0.3,-0.1)$);
      \end{tikzpicture}
      };
      \node[below=-7pt of a] (f2y3) {
      \begin{tikzpicture}
    \matrix (m) [matrix of math nodes, row sep=0.7em,ampersand replacement=\&,
      column sep=1em]{
      X_p+Y\& X_pY \& X+Y_p \\
      XT+Y_p\& XTY_p \& XT+Y \\};
    \path[-stealth]
      (m-1-2) edge (m-1-3)
      ($(m-1-1)+(0.6,0.1)$) edge ($(m-1-2)+(-0.4,0.1)$)
      ($(m-1-2)+(-0.4,-0.1)$) edge ($(m-1-1)+(0.6,-0.1)$)
      (m-2-2) edge (m-2-3)
      ($(m-2-1)+(0.6,0.1)$) edge ($(m-2-2)+(-0.4,0.1)$)
      ($(m-2-2)+(-0.4,-0.1)$) edge ($(m-2-1)+(0.6,-0.1)$);
      \end{tikzpicture}
      };
    \end{tikzpicture}
  };
  \node[left=-15pt of a] (dibu) {
    \begin{tikzpicture}[node distance=0.2cm]
      \node[inner sep = 0pt] at(0.6,-0.9) (a1) {};
      \node[inner sep = 0pt] at(1.25,-0.9) (a2) {};
      \node[draw,circle,text width=0.3cm]  at (0,0) (xt) {\textbf X};
      \node[below=0mm of xt,draw, rectangle,text width=0.13cm] (t) {\hspace{-0.8mm}\tiny \textbf{T}};
      \node[left=5mm of xt, draw,circle,text width=0.3cm]  (x) {\textbf X};
      \node[right=0.7cm of t,draw, circle,text width=0.3cm,fill=gray!30] (yp) {\textbf{Y$_p$}};
      \node[below=of yp,draw, circle,text width=0.3cm] (y) {\textbf{Y}};
      \node[right=1.8cm of xt, draw,circle,text width=0.3cm]  (xp) {\textbf{X$_p$}};
      \draw (y) edge[->, bend right=45,thick] node[above] {} (yp);
      \draw (y) edge[<-, bend left=45,thick]  node[above] {} (yp);
      \draw (xp) edge[<-,bend right=45,thick]  node[above] {} ($(xt.east)+(-1pt,7pt)$);
      \draw (a1) edge[bend left=25,thick]  node[above] {} ($(t.east)+(-0.4pt,7pt)$);
      \draw (a2) edge[bend right=25,thick]  node[above] {} ($(xp.south)+(-6pt,1.4pt)$);
      \filldraw(a2) circle (.04cm) ;
      \filldraw(a1) circle (.04cm) ;
      \draw (a1) edge [->,bend left=15,thick] ($(a1)+(-0.27,-0.3)$) node[right] {};
      \draw[shift={($(a1)+(-0.3,-0.3)$)}]  node[below left=-6pt] {\tiny \textbf{P$_i$}};
      \draw (1,0.7) edge [->,bend right=15,thick] (1.7,0.9) node[above] {};
      \draw[shift={(1.7,0.9)}]  node[above right=-2pt] {\tiny \textbf{D}};
      \draw[<->,thick] (x.east)--(xt.west);
    \end{tikzpicture}
  };
\end{tikzpicture}}\\
{\footnotesize (C)} & {\footnotesize (D)} 
\end{tabular}
\caption{Examples of MESSI systems: Sequential n-site phosphorylation/dephosphorlation 
(A) distributive case \cite{PM12,ws08},
(B) processive case \cite{ShCo,MaHoKh04};
(C) Phosphorylation cascade;
(D) Schematic diagram of an EnvZ-OmpR bacterial model~\cite{sf10}.}\label{fig:4}
\end{figure}

We depict in Figures~\ref{fig:4} and~\ref{fig:sadi} 
some examples of important biochemical networks which are MESSI networks.
\footnote{As usual, in the figures we summarize with the scheme
{\small
\begin{tikzpicture}[node distance=0.5cm]
  \node[] at (1.1,-1.5) (dummy2) {};
  \node[left=of dummy2] (p0) {\textcolor{black!70}{$S_0$}};
  \node[right=of p0] (p1) {\textcolor{black!70}{$S_1$}}
    edge[<-, bend right=45,color=black!60] node[above] {\textcolor{black!60}{\small{${E}$}}} (p0);
\end{tikzpicture}
}
a sequence of reactions with intermediates such as
$S_0+E \overset{\kappa_1}{\underset{\kappa_2}{\rightleftarrows}}
  ES_0 \overset{\kappa_3}{\rightarrow}
  S_1+E$.}
Figure~\ref{fig:4}(A) features the $n$-site phosphorylation-de\-phos\-phorylation of a protein by a 
kinase-phosphatase pair in a sequential and distributive mechanism. The total of $n$ phosphate groups 
are allowed to be added to the unphosphorylated substrate $S_0$ by an enzyme $E$.
The substrate $S_i$ is the phosphoform obtained from $S_0$ by attaching $i$ phosphate groups to it.  
Each phosphoform can accept (via an enzymatic reaction involving $E$) or lose (via a reaction 
involving the phosphatase $F$) at most one phosphate; this means that the mechanism is ``distributive.''
In addition,  the phosphorylation is said to be ``sequential'' because multiple phosphate groups must 
be added in a specific order and removed in a specific order as well.
The sequential and processive phosphorylation/de\-phos\-phorylation 
of a substrate at $n$ sites \cite{MaHoKh04,ShCo} is depicted in Figure~\ref{fig:4}(B).
The substrate undergoes $n \geq 1$ phosphorylations after binding to the kinase and forming the 
enzyme-substrate complex; only the fully phosphorylated substrate is released, and hence only 
two phosphoforms have to be considered: the unphosphorylated substrate $S_0$ and the fully phosphorylated substrate $S_n$. 
Processive dephosphorylation proceeds similarly.
All the motifs in \cite{fw12} are MESSI networks, as are the phosphorylation cascades shown in Figure~\ref{fig:sadi}.  
The cascade in Figure~\ref{fig:4}(C) features the sequential activation of a
specific MAPK kinase kinase (MAPKKK, denoted $S$) and a MAPK kinase (MAPKK, denoted $P$), which in turn
phosphorylates and activates the downstream MAPK (denoted $R$). The activated forms are
$S_1$, $P_2$ and $R_2$, respectively. Figure~\ref{fig:sadi} features two cascade motifs with two layers,
which are a combination of two one-site modification cycles with either a specific or the same
phosphatase acting in each layer. It is already known \cite{fw12} 
that the cascade in (A) exhibits multistationarity while the
cascade in (B) is monostationary. We will recover these results under the framework of MESSI
systems (they will both prove to be s-toric MESSI systems, see Definition~\ref{def:storic}). We will moreover consider the
cascade in Figure~\ref{fig:sadi} (A) as one of our running examples in this article, and sometimes
we will also include a drug $D$ acting by a sequestration mechanism such as $P_1+D \rightleftarrows P_1D$.
{\small
\begin{figure}
\centering
\begin{tabular}{ccc}
 {\footnotesize (A)} \begin{tikzpicture}[scale=0.7,node distance=0.5cm]
  \node[] at (1.1,-1.5) (dummy2) {};
  \node[left=of dummy2] (p0) {$P_0$};
  \node[right=of p0] (p1) {$P_1$}
    edge[->, bend left=45] node[below] {\textcolor{black!70}{\small{$F$}}} (p0)
    edge[<-, bend right=45] node[above] {} (p0);
  \node[] at (-1,0) (dummy) {};
  \node[left=of dummy] (s0) {$S_0$};
  \node[right=of s0] (s1) {\textcolor{blue}{$S_1$}}
    edge[->, bend left=45] node[below] {\textcolor{black!70}{\small{$F$}}} (s0)
    edge[<-, bend right=45] node[above] {\textcolor{black!70}{\small{$E$}}} (s0)
    edge[->,thick,color=blue, bend left=25] node[above] {} ($(p0.north)+(25pt,10pt)$);
 \end{tikzpicture}
& \hspace{0.6cm} &
 {\footnotesize (B)} \begin{tikzpicture}[scale=0.7,node distance=0.5cm]
  \node[] at (1.1,-1.5) (dummy2) {};
  \node[left=of dummy2] (p0) {$P_0$};
  \node[right=of p0] (p1) {$P_1$}
    edge[->, bend left=45] node[below] {\textcolor{black!70}{\small{$F_2$}}} (p0)
    edge[<-, bend right=45] node[above] {} (p0);
  \node[] at (-1,0) (dummy) {};
  \node[left=of dummy] (s0) {$S_0$};
  \node[right=of s0] (s1) {\textcolor{blue}{$S_1$}}
    edge[->, bend left=45] node[below] {\textcolor{black!70}{\small{$F_1$}}} (s0)
    edge[<-, bend right=45] node[above] {\textcolor{black!70}{\small{$E$}}} (s0)
    edge[->,thick,color=blue, bend left=25] node[above] {} ($(p0.north)+(25pt,10pt)$);
 \end{tikzpicture}
\end{tabular}
\caption{Same and different phosphatases in two different MESSI cascades}\label{fig:sadi}
\end{figure}
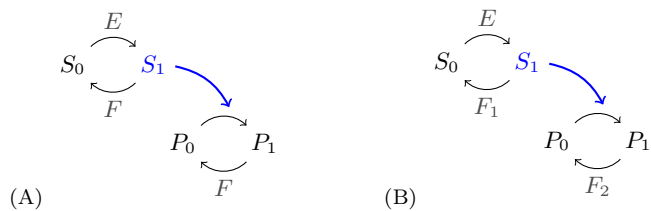}
Figure~\ref{fig:4}(D) depicts a schematic diagram of an EnvZ-OmpR bacterial model~\cite{sf10}, which
is a MESSI network.
The sensor EnvZ (X) phosphorylates itself by binding and breaking
down ATP (T). The phosphorylated form X$_p$ catalyzes the transfer
of a phosphoryl group to the response-regulator OmpR (Y). X,
together with ATP dephosphorylates Y$_p$, a transcription factor that regulates
the expression of various protein pores.

Our work continues the ideas in chemical reaction network theory (CNRT), 
which connects qualitative properties of ordinary differential equations  
corresponding to a reaction network to the network structure.
CNRT has been developed over the last 40 years, initially through the 
work of Horn and Jackson and subsequently by Feinberg 
and his students and collaborators (for example, see~\cite{fe79,
Fein95DefOne})  and Vol'pert \cite{vh85}.
Biochemical reaction networks, that is, chemical reaction networks in biochemistry, 
is the principal current application of these developments. In particular, our work 
is inspired by previous articles by Thomson and Gunawardena \cite{TG09}, who set 
the posttranslational modification (PTM) framework; Mirzaev and Gunawardena \cite{MiGu13}, 
who detailed the Laplacian dynamics; Feliu and Wiuf \cite{fw13,fw13b}, who clarified the 
elimination of intermediate complexes; and M\"{u}ller et al. \cite{mfrcsd13}, who collected and 
clarified the role of signs in the determination of multistationarity. 
Also related to our work are the papers by Gnacadja on constructive chemical reaction 
networks~\cite{Gn11a,Gn11b}, who gave an alternative approach to the PTM setting.
The MESSI structure we propose simplifies and unifies most of these approaches.

The precise conditions are given in Definitions~\ref{def:messi} and~\ref{def:ms}.
In particular,  complexes in a MESSI network are mono or bimolecular. As remarked in~\cite{TG09}, 
one main assumption for this modeling is that donor molecules
that provide modifiers are kept at constant concentration on the time scaling of the
reactions we are modeling, and their effects can be absorbed into the rate constants.
The main difference between our approach and theirs is that they do not allow  a
species to act as a substrate in one reaction and then as an enzyme in another 
(neither does~\cite{MiGu13}), which in particular excludes all enzymatic cascades. 
This is considered in~\cite{Gn11a,Gn11b}. However, none of these previous settings
allow swaps and monomolecular reactions between core species that our framework incorporates. 
Regarding~\cite{fw13,fw13b,mfrcsd13}, we pay special attention to networks with toric steady 
states~\cite{PM12}.

Theorem~\ref{th:conservations} explicitly describes conservation relations that imply that any 
MESSI system is conservative. Theorem~\ref{th:bss} gives conditions that ensure that a MESSI 
system is persistent. We give necessary conditions for the existence of a rational parametrization 
of the variety of positive steady states in Theorem~\ref{th:par}, which is the generalization
of the main theorem in~\cite{TG09} to our setting.
Proposition~\ref{prop:int} expresses the role of intermediates in the steady states of the system.  
Theorem~\ref{th:toric_toric} shows a frequent class of MESSI systems with special steady states, 
cut out by binomial equations and termed as toric steady states~\cite{PM12}, 
that allow for an easier determination of multistationarity. 

We give for MESSI systems with toric steady states an algorithm to determine the 
capacity for multistationarity based on Theorems~\ref{th:monostationarity} 
and~\ref{th:multistationarity}.
If this is the case, the algorithm provides rate constants for which multistationarity 
takes place, based on the theory of oriented matroids~\cite{Bjetal99}. This is a 
specialized procedure, easy to tune to produce different choices of rate constants, 
besides the general algorithms for injectivity implemented, for instance, by Feinberg 
and his group in the Chemical Reaction Network Toolbox~\cite{Tool}. Links to other 
algorithms can be found at 
{\verb#https://reaction-networks.net/wiki/Mathematics_of_Reaction_Networks#}.
The proofs of our statements are concentrated in the Appendix.

\bigskip

 
\section{MESSI systems}

\medskip

In this section we review the notion of a chemical reaction network in order to introduce the definition of 
MESSI networks and MESSI systems (when these networks are endowed with mass-action kinetics). 
The conditions in the definition might seem to be very restrictive (mathematically), but indeed 
we show many examples of popular  networks in systems biology that lie in this framework.

\subsection*{Chemical reaction systems}

\medskip

\noindent We briefly recall the basic setup of chemical reaction networks and how they give rise to autonomous 
dynamical systems under mass-action kinetics (see Example~\ref{ex:enzsys}). 
Given a set of $s$ chemical species, a {\em chemical reaction network} on this set of species is a finite 
directed graph whose vertices are indicated by complexes and whose edges are labeled by parameters 
(reaction rate constants). The labeled digraph is denoted $G = (\Ve,\Edg, \Ka)$, with vertex set $\Ve$,
edge set $\,\Edg$, and edge labels $\Ka \in \R_{>0}^{\#\Edg}$.
If $(y,y')\in \Edg$, we denote $y\to y'$. Complexes determine  vectors in $\Z_{\ge 0}^{s}$
according to the stoichiometry of the species they consist of.
We identify a complex  with its corresponding vector and also with the formal linear
combination of species specified by its coordinates. 

\begin{example}[Basic example of an enzymatic network] \label{ex:enzsys}
We present a basic example that illustrates how a chemical reaction network gives 
rise to a dynamical system. This example represents a classical mechanism of enzymatic 
reactions, usually known as the futile cycle \cite{sig-016,kholo00,ws08}:
\begin{equation}\label{eq:reaction_ex}
 S_0+E \underset{\kappa_2}{\overset{\kappa_1}{\rightleftarrows}} U_1 \overset{\kappa_3}{\rightarrow} S_1+E \qquad
 S_1+F \underset{\kappa_5}{\overset{\kappa_4}{\rightleftarrows}} U_2 \overset{\kappa_6}{\rightarrow} S_0+F,
\end{equation}
where $U_1$ and $U_2$ are \emph{intermediate species},
$S_0$ and $S_1$ are \emph{substrates}, and 
$E$ and $F$ are \emph{enzymes}.
The source and the product of each reaction are called {\em complexes}.
The concentrations of the six species change in time as the reactions occur.
We order the $s=6$ species as follows: $U_1$, $U_2$, $S_0$, $S_1$, $E$, $F$, and we denote the 
concentrations by
$[U_1]=u_1$, $[U_2]=u_2$, $[S_0]=x_1$, $[S_1]=x_2$, $[E]=x_3$,  $[F]=x_4$.
The first three complexes in the network~\eqref{eq:reaction_ex} give rise to
the vectors $(0,0,1,0,1,0)$, $(1,0,0,0,0,0)$, and $(0,0,0,1,1,0)$.
Under the assumption of {\em mass-action kinetics},
we obtain then the following polynomial dynamical
system:
\[
\begin{array}{r@{\hskip 1mm}c@{\hskip 1mm}cc@{\hskip 1mm}c@{\hskip 1mm}l}
\frac{du_1}{dt} & = &\kappa_1 x_1.x_3-(\kappa_2+\kappa_3)u_1,
& \frac{du_2}{dt} & = &\kappa_4 x_2.x_4-(\kappa_5+\kappa_6)u_2,\\
\\
\frac{dx_1}{dt} & = &-\kappa_1 x_1.x_3+\kappa_2 u_1+\kappa_6 u_2,
& \frac{dx_2}{dt} & = &-\kappa_4 x_2.x_4+\kappa_5 u_2+\kappa_3 u_1,\\
\\
\frac{dx_3}{dt} & = &-\kappa_1 x_1.x_3+(\kappa_2+\kappa_3)u_1,
& \frac{dx_4}{dt} & = &-\kappa_4 x_2.x_4+(\kappa_4+\kappa_5)u_2.
\end{array}\]
\end{example}

The unknowns $x_1,x_2,\ldots,x_s$ represent the
concentrations of the species in the network, and we regard them as functions of time $t$.
Under mass-action kinetics, the chemical reaction network $G$ defines the
following chemical reaction dynamical system:
\begin{equation}
\label{CRN}
{\dot \bx}~=~\left( \frac{dx_1}{dt} ,\frac{dx_2}{dt}  ,\dots , \frac{dx_s}{dt}   \right)  ~=~
\underset{y\to y'}{\sum} \kappa_{yy'} \,  \bx^y \, (y'-y),
\end{equation}
\noindent where $\bx=(x_1,\dots,x_s)$ and $\bx^y=x_1^{y_1}\cdots x_s^{y_s}$.
The right-hand side of each differential equation $dx_\ell/dt$ is a polynomial 
$f_\ell(\bx,\Ka)$, in the variables $x_1,\dots, x_s$ with real coefficients  $\Ka$.
The associated \emph{steady state variety} $V_f$
is defined as the common nonnegative zeros of the polynomials $f_\ell$, that is, 
\begin{equation}\label{eq:Vf}
 V_f  := \{ x \in \R_{\ge 0}^ s \, : \,  f_\ell(x, \Ka) =0, \quad \ell=1,\dots, s\}.
\end{equation}

The linear subspace spanned by the reaction vectors $\stoich = \{ y' -y\, : \, y \to y' \}$ is called the
{\em stoichiometric subspace}.
Notice from~\eqref{CRN} that the vector $\dot \bx(t)$  lies in $\stoich$ for all time $t$.
In fact, a trajectory $\bx(t)$ beginning at a vector $\bx(0)=\bx^0 \in \R^s_{\ge 0}$
remains in the {\em stoichiometric compatibility class}
$(\bx^0+\stoich) \cap \mathbb{R}^s_{\geq 0}$ for all positive time. The equations of
$\bx^0+\stoich$ give rise to linear \emph{conservation relations} of the system. 

\begin{definition}\label{def:multi}
We say that the system {\em has the capacity for multistationarity}
if there exists a choice of rate constants $\Ka$ such that there are two or more steady
states in one stoichiometric compatibility class. 
On the other hand, if for any choice of rate constants 
there is at most one steady state in each stoichiometric compatibility class,
the system is said to be \emph{monostationary}.
\end{definition} 
It may happen that the vectors $\dot{\bx}(t)$ lie in a smaller subspace $\stoich'\subseteq \stoich$,
called the \textit{kinetic subspace} \cite{feho77}. In this case,  the trajectories live
in $(x^0+\stoich') \cap \mathbb{R}^s_{\geq 0}$ for some initial state $x^0 \in \mathbb{R}^s_{\geq 0}$,
and the concepts of mono- and multistationarity might be defined with respect to this smaller affine subspace.
In this article, we focus on the classical Definition~\ref{def:multi}.

\subsection*{Definition of MESSI systems}

\medskip

A MESSI network is a particular type of chemical reaction network, which includes all monomolecular (linear) ones. 
As we mentioned in the introduction, the main ingredient in the definition is the existence of a {partition} of the 
set of species that is, a decomposition into disjoint subsets, with the following properties.

\begin{definition}\label{def:messi}
A chemical reaction network is called a MESSI network if there is a {\bf partition} of the set of {species} $\Sp$ 
 \begin{equation}\label{eq:S}
\Sp=\Sp^{(0)}\bigsqcup \Sp^{(1)} \bigsqcup \Sp^{(2)} \bigsqcup \dots \bigsqcup \Sp^{(m)},
\end{equation}
where $m \ge 1$ and $\bigsqcup$ denotes disjoint union, such that the {complexes} and {reactions} satisfy the 
conditions below.

We call the cardinalities $\#\Sp^{(0)}=p$, $\#\Sp^{(\alpha)}=n_\alpha$ for any $\alpha > 0$ and
$\underset{\alpha >0}{\sum}n_\alpha=n$. We allow $p$ to be $0$, but we assume that all $n_\alpha$ are
positive.  Species in $\Sp^{(0)}$ are called {\bf intermediate},
and species in $\Sp_1:= \Sp \setminus\Sp^{(0)}$ are termed {\bf core}.
When convenient, we will distinguish intermediate and core species in the notation in the following way:
$\Sp^{(0)}=\{U_1,\dots, U_p\}$, $\Sp_1 =\{X_1,\dots,X_n\}$. Thus, the vectors determined 
by the {complexes} $(\lambda_1, \dots, \lambda_p, \nu_1,\dots, \nu_n)$  live in $\Z_{\ge 0}^{p+n}$ 
and define the formal linear combination of species $\sum_{i=1}^p \lambda_i U_i + \sum_{j=1}^n \nu_j X_j$. 

\medskip

Complexes are also partitioned into two disjoint sets, and the following conditions hold:
\begin{itemize}
\item[$(\mathcal N_1)$] {\bf Intermediate complexes} are complexes that
consist of a unique intermediate species that only appears in that complex.
The vector corresponding to the unimolecular complex $U_i$ is denoted by $y_i$.
\item[$(\mathcal N_2)$] {\bf Core complexes} \cite{fw13} are mono or bimolecular and consist of either one
or two core species. If the core complex consists of only the species $X_i$,
the corresponding vector will  be denoted by $y_{i0}$.  
\item[$(\mathcal N_3)$] When a core complex consists of two species
$X_i, X_j$, they {\em must} belong to \emph{different} sets $\Sp^{(\alpha)}, \Sp^{(\beta)}$
with $\alpha \neq \beta, \alpha, \beta \geq1$. We also denote the complex $X_i+ X_j = X_j+ X_i$ by $y_{ij}=y_{ji}$.
\end{itemize}

We say that complex $y$  reacts to complex $y'$ {\bf via intermediates} if either $y\to y'$ or
there exists a path of reactions from $y$ to $y'$ {\em only through intermediate complexes}.  
This is denoted  by $y \uri y'$.
The intermediate complexes of a MESSI network satisfy, moreover, the following condition:
\begin{itemize}
\item[$(\cond)$]  For every intermediate complex $y_k$, there exist core
complexes $y_{ij}$ and $y_{\ell m}$ such that $y_{ij}\uri y_k$ and $y_k\uri y_{\ell m}$.
\end{itemize}

\medskip

Finally, {reactions} are constrained by the following rules:
\begin{itemize}
 \item[$({\mathcal R}_1)$]  If three species are related by $X_i+ X_j \uri X_k$ or $X_k \uri X_i + X_j$, 
 then $X_k$ is an intermediate species.
 \item[$({\mathcal R}_2)$]  If two core species $X_i, X_j$ are related by $X_i\uri X_j$, then there exists 
 $\alpha \ge 1$ such that both belong to $\Sp^{(\alpha)}$.
 \item[$({\mathcal R}_3)$] If $X_i+X_j\uri X_k+X_\ell$, then there exist $\alpha \neq \beta$ such that 
 $X_i,X_k \in \Sp^{(\alpha)}$, $X_j,X_\ell \in \Sp^{(\beta)}$ or $X_i,X_\ell \in \Sp^{(\alpha)}$, 
 $X_j,X_k \in \Sp^{(\beta)}$.
\end{itemize}

We will say that the partition~\eqref{eq:S} \emph{defines a MESSI structure} on
the network.
\end{definition}

\begin{example} \label{ex:toy}
We present a toy example that shows which kinds of reactions are allowed and which are not.
Consider the following digraph, where we assume $Y_1$ and $Y_2$ to be monomolecular complexes:
\[X_1+X_2\to Y_1 \rightleftarrows Y_2 \to Y_3.\]
Then, $Y_1$ and $Y_2$ must consist of an intermediate species by rule $({\mathcal R}_1)$.
For  Condition~($\cond$) to hold, necessarily $Y_3$ must be a core
complex since there are no arrows leaving from $Y_3$. Moreover, rule $({\mathcal R}_1)$ imposes
that $Y_3$ is of the form $X_\ell+X_m$, and by rule $({\mathcal R}_3)$, if $X_1\in \Sp^{(\alpha)}$ and
$X_2 \in \Sp^{(\beta)}$, then $\alpha\neq\beta$ and either $X_\ell\in\Sp^{(\alpha)},
X_m\in \Sp^{(\beta)}$ or $X_m\in\Sp^{(\alpha)},
X_\ell\in \Sp^{(\beta)}$.
\end{example}


Notice that a MESSI network is defined once the partition of $\Sp$ is given and all conditions and 
rules in Definition~\ref{def:messi}  are verified.
It is important to point out that  even if in the chemical setting there are natural partitions 
of the set of species given by the different types of molecules, there can be many ways to define a
partition which defines a MESSI structure. 
We can define a partial order in the set of all possible partitions of the species of a given biochemical network.  
\begin{definition}\label{def:part}
Given two partitions
$\Sp = \Sp^{(0)}\sqcup \Sp^{(1)} \sqcup \Sp^{(2)} \sqcup \dots \bigsqcup \Sp^{(m)}$ 
and $\Sp = {\Sp'}^{(0)}\sqcup \Sp'^{(1)} \sqcup \Sp'^{(2)} \sqcup \dots \bigsqcup \Sp'^{(m')}$, 
we say that the first partition refines the second one if and only if $\Sp^{(0)}\supseteq {\Sp'}^{(0)}$ 
and for any $\alpha \ge 1$, there exists
$\alpha'\ge 1$ such that $\Sp^{(\alpha)}\subseteq {\Sp'}^{(\alpha')}$.
With this partial order we have the notion of  a {minimal partition}. 
\end{definition}

Before presenting our two running examples, we define enzyme behavior and swaps.
\begin{definition}\label{def:enzyme}
A species $X_j$ that satisfies $X_i+X_j\uri X_\ell+X_j$ for some $X_i, X_\ell$ is said
to act as an \emph{enzyme}. In this case, we call $X_i$ the \emph{substrate} 
and $X_\ell$ the \emph{product}.
A reaction via intermediates is called a \emph{swap} if $X_i+X_j \uri X_\ell+X_m$,
and $i,j\notin\{\ell,m\}$ (so, neither $X_i$ nor $X_j$ acts as an enzyme
in $X_i+X_j \uri X_\ell+X_m$).
\end{definition}

Notice that if a species $X_j$ in a MESSI network \emph{only} acts as an enzyme,
we can consider a singleton subset $\Sp^{(\alpha)}=\{X_j\}$.

\begin{example}[First running example]\label{ex:ES}
 Consider the network in Figure~\ref{fig:sadi} (A), with digraph 
{\small
  $$S_0+E \overset{\kappa_1}{\underset{\kappa_2}{\rightleftarrows}}
  ES_0 \overset{\kappa_3}{\rightarrow}
  S_1+E  \quad \quad
  S_1+F \overset{\kappa_{4}}{\underset{\kappa_{5}}{\rightleftarrows}}
  FS_1 \overset{\kappa_{6}}{\rightarrow}
  S_0+F$$
  $$P_0+S_1 \overset{\kappa_7}{\underset{\kappa_8}{\rightleftarrows}}
  S_1P_0 \overset{\kappa_9}{\rightarrow}
  P_1+S_1  \quad \quad
  P_1+F \overset{\kappa_{10}}{\underset{\kappa_{11}}{\rightleftarrows}}
  FP_1 \overset{\kappa_{12}}{\rightarrow}
  P_0+F.$$}
We can consider the partition $\Sp^{(0)}=\{ES_0,FS_1,S_1P_0, FP_1\}$  (intermediate species), and
$\Sp^{(1)}=\{S_0,S_1\}$, $\Sp^{(2)}=\{P_0,P_1\}$, $\Sp^{(3)}=\{E\}$, $\Sp^{(4)}=\{F\}$ (partition of the core species).
The intermediate complexes correspond to the intermediate species, and the remaining complexes are core complexes. 
This partition defines a MESSI structure in the network.
In fact, there is another possible choice
of partition which  also gives a MESSI structure to the network, considering
$\Sp^{(0)}$, $\Sp^{(1)}$ and $\Sp^{(2)}$ as before, but $\Sp^{(3)}$ and $\Sp^{(4)}$ are replaced
by their union $\{E,F\}$. We can see in this example that species $E$ and $F$ 
only act as enzymes, while species $S_1$ acts as an enzyme in the second layer but in the first 
one it plays the role of a substrate of  $F$ and of a product of $E$.
\end{example}

\begin{example}[Second running example]\label{ex:S}
An example of swap can be the seen in the transfer of a modifier molecule,
such as a phosphate group in a two-component system, from one molecule to another.
We consider as our second running example the EnvZ/OmpR system. The corresponding digraph
$G$ is featured in Figure~\ref{fig:4}(D).
The only possible partition for this network to be a MESSI network is 
$\Sp^{(0)}=\{\text{X}_p\text{Y}, \text{XTY}_p\}$, $\Sp^{(1)}=\{X,XT,X_p\}$, $\Sp^{(2)}=\{Y,Y_p\}$.
The reaction via intermediates in the second connected component of the graph of reactions is a swap.  
On the other hand, $XT$ acts as an enzyme in the last component of $G$. 
\end{example}

In Example~\ref{ex:ES}, there are two different partitions, but the first one is a refinement of the second one. 
However, there might be noncomparable partitions, as we show in the following example.

\begin{example}[Non-comparable partitions]\label{ex:noref} 
Consider the following network:
\[X_1+X_2 \to X_3+X_4 , \quad  X_4+X_5 \to X_6+X_1.\]
Set $\Sp^{(0)}=\emptyset$, $\Sp^{(1)}=\{X_1,X_4\}$ and $\Sp^{(2)}=\{X_2,X_3,X_5,X_6\}$. We can refine
$\Sp^{(2)}$ into $\Sp'^{(2)}=\{X_2,X_3\}$ and $\Sp'^{(3)}=\{X_5,X_6\}$.  
In both cases, we get the structure of a MESSI network. If we instead consider
$\Sp''^{(0)}=\emptyset$, $\Sp''^{(1)}=\{X_1,X_3,X_5\}$ and $\Sp''^{(2)}=\{X_2,X_4,X_6\}$ 
there is no possible way of refining $\Sp''^{(2)}$ without violating $({\mathcal R}_3)$. 
The second and third partitions are not comparable, and both are minimal in the poset of
partitions of the species set which yield a MESSI structure on the given network.
\end{example}

The main focus of this work is the properties of MESSI networks endowed with 
kinetics. Throughout this text we will always assume mass-action kinetics.
\begin{definition}\label{def:ms}
 We call a \emph{MESSI system} the mass-action kinetics dynamical system as
 in~\eqref{CRN} associated with a MESSI network.
\end{definition}


\section{Conservation relations and persistence in MESSI systems}\label{sec:cons}

\medskip

We first describe the equations of the stoichiometric subspace of a MESSI system, which give linear conservation 
relations along the trajectories. We then focus on the steady states of MESSI systems.
We  give sufficient conditions for MESSI systems to be persistent.

\subsection*{Conservation relations}

\medskip

\noindent A chemical reaction system  is said to be \emph{conservative} if there exists a linear 
combination of the species in the network with all \emph{positive} coefficients which is constant
along each trajectory (i.e., for all time $t$).
Clearly, for any trajectory starting at a positive point, this constant is a positive real number. 
In this case, all stoichiometric compatibility classes are compact.
In this section we show that MESSI systems are conservative, by exhibiting natural conservation relations. 
This implies that all trajectories are bounded and defined for any positive time. 

\begin{notation} We denote the concentration of the species with small letters.
For example, $u_i$ denotes the concentration of $U_i$ and $x_j$ denotes the concentration of $X_j$.
\end{notation}

Given a  MESSI network and a partition of the species set as in Definition~\ref{def:messi}, 
we define  for any $\alpha \ge 1$ the set of indices
\begin{equation}\label{eq:intal}
 \intal(\alpha)=\{k: \text{there exists } y_{ij} \text{ with either }
X_i\in \Sp^{(\alpha)} \text{ or } X_j\in \Sp^{(\alpha)} \text{ such that } y_{ij}  \uri y_k\}.
\end{equation}
We also denote by $\Sp\intal(\alpha)$ the set of species with indices in $\intal(\alpha)$.
Note that the subsets $\intal(\alpha)$ are in general not disjoint, but condition $(\mathcal C)$ implies
that $\cup_{\alpha\ge1} \Sp\intal (\alpha) = \Sp^{(0)}$.
It is straightforward to see that the conditions imposed on a MESSI 
network ensure that for any $\alpha \ge 1$ the set of variables 
$\Sp^{(\alpha)} \cup \Sp\intal(\alpha)$ is a siphon~\cite{ADLS07}.
We will show in Theorem~\ref{th:conservations} below that the following explicit linear conservation 
relations with $\{0,1\}$ coefficients hold:
\begin{equation}\label{eq:consalpha}
\ell_\alpha(u,x)\, = \, \cte_\alpha, \text{ where } 
\ \ell_\alpha(u,x)= \sum_{X_i \in \Sp^{(\alpha)}} x_i + \sum_{k\in\intal(\alpha)} u_k,
\end{equation}
for some  constant $\cte_\alpha$, which is positive if the trajectory intersects the positive orthant. 
This is a direct consequence of Theorem~2.1 in~\cite{fw13} and of Theorem~5.3 in \cite{Gn11b}.  
The second part of Theorem~\ref{th:conservations} gives sufficient conditions for these  relations 
to generate \emph{all} the equations defining a stoichiometric compatibility class.  
We show in Example~\ref{ex:22} that if we relax any of these conditions, 
the result is not true. See also Proposition~\ref{prop:kinstoi} on the conditions to ensure 
that the kinetic and the stoichiometric subspaces coincide. 

\begin{theorem}\label{th:conservations}
 Given a chemical reaction network $G$ and a partition of the set of species $\Sp$ 
 as in~\eqref{eq:S} that defines a MESSI structure, for each subset of species 
 $\Sp^{(\alpha)}$, $1\le\alpha \le m$, the linear form $\ell_\alpha$ in~\eqref{eq:consalpha} 
 defines a conservation relation of the system. 
 In particular, all MESSI systems are conservative.
 
 Furthermore, if there are no swaps in $G$, and the partition is \emph{minimal} in the poset 
 of partitions defining a MESSI system structure on $G$, then $\dim(\stoich^\bot)=m$.
 
 If, moreover, the stoichiometric subspace coincides with the kinetic subspace, then 
 the only possible conservation relations in the system are linearly generated by 
 the conservations~\eqref{eq:consalpha} for $1\leq \alpha\leq m$.
\end{theorem}

\begin{example}[Examples~\ref{ex:ES} and \ref{ex:S}, continued]
 For the cascade with one phosphatase in Example~\ref{ex:ES}, 
the hypotheses in Theorem~\ref{th:conservations} are satisfied and the conservation relations are the following:
 \begin{align*}
  s_0+s_1+u_1+u_2+u_3= S_{tot},\quad &
  p_0+p_1+u_3+u_4= P_{tot},\\
  e+u_1=  E_{tot},\quad &
  f+u_2+u_4=  F_{tot},
 \end{align*}
where we use small letters for the concentration of the corresponding species. 
The concentration of the intermediates species $es_0, fs_1,s_1p_0,fp_1$ are denoted by $u_1, u_2, u_3, u_4$, 
respec\-tively.
In Example~\ref{ex:S}, the conservation relations are
 \begin{equation*}
  x+xt+x_p+x_py+xty_p= X_{tot},\quad
  y+y_p+x_py+xty_p= Y_{tot}.
 \end{equation*}
\end{example}

\begin{example}[Necessity of the hypotheses in Theorem~\ref{th:conservations}]\label{ex:22}
The following is Example 22 from \cite{SaWiFe15}. It satisfies the hypotheses in Theorem~\ref{th:conservations}
except for the absence of swaps:
\[
\begin{array}{l}
X_1+X_5 \rightarrow X_2+X_6\\
X_3+X_6 \rightarrow X_4+X_5\\
X_4+X_6 \rightarrow X_3+X_7.
\end{array}
\]
It is straightforward to see that the only possible {\em minimal} partition
is $\Sp^{(1)}=\{X_1,X_2\}$, $\Sp^{(2)}=\{X_3,X_4\}$, $\Sp^{(3)}=\{X_5,X_6,X_7\}$,
which gives three linearly independent conservation relations $\ell_1, \ell_2, \ell_3$. 
However, there is a fourth independent conservation relation:
\[x_1 +x_4 +x_6 + 2 x_7 = \cte.\]
 \end{example}

Before stating the sufficient conditions to ensure that the kinetic and the stoichiometric subspaces coincide,
we recall some concepts from graph theory that will be useful in the rest of the article.

Given a directed graph $G= (\mathcal V, \mathcal E)$, define the following equivalence relation
between the vertices: two vertices $i, j \in \mathcal V$ are related if and only if there is a
directed path from $i$ to $j$, and a directed path from $j$ to $i$. Equivalence classes of vertices
define the vertices of the \emph{strongly connected} components of $G$. Thus, a directed graph is  
strongly connected when for each ordered pair of vertices there is a directed path from the first 
vertex to the second one. Note that the underlying undirected graph of a strongly connected graph 
is connected. If one strongly connected component has no edges from any node in the component to a 
node in a different strongly connected component, it is called a \textit{terminal strongly connected component}.

A directed graph $G$ is said to be \emph{weakly reversible} if each connected component is strongly connected. 
This means that if there is a directed path from a vertex $i$ to another vertex $j$, there is also a directed 
path from $j$ to $i$, but it could happen that no path exists in any of the two directions.
Thus $G$ is strongly connected if and only if it is weakly reversible and connected, and the connected
components of a weakly reversible graph are strongly connected. 

\begin{example}\label{ex:wrsc}
The underlying directed graph of the chemical reaction network
\begin{equation*}
X_3 \overset{\kappa_1}{\leftarrow} X_1 \overset{\kappa_2}{\underset{\kappa_3}{\rightleftarrows}}X_2 \overset{\kappa_4}{\rightarrow} X_4,
\end{equation*}
is connected but not weakly reversible. It has three strongly connected components: the node $X_3$ (with no arrows),
the node $X_4$ (again, with no arrows),  which are terminal strongly connected components, 
and the subgraph $X_1 {\rightleftarrows} X_2$, which is not terminal.
\end{example}

The following result is from~\cite{feho77}.

\begin{proposition} \label{prop:kinstoi} 
If $G$ has only one terminal strongly connected component in each connected component,
the number of generators of the conservation relations is $s-\dim(\stoich)$,
where $s$ is the total number of species and $\stoich$ is the stoichiometric subspace.
In this case, the stoichiometric and the kinetic subspaces coincide.
\end{proposition}

When there is more than one terminal strongly connected component in one connected component,
even if there are no swaps, we can find other conservation relations. For instance, consider the chemical
reaction network  in Example~\ref{ex:wrsc} and the partition of the
set of species: $\Sp^{(0)}=\emptyset$ and $\Sp^{(1)}=\{X_1,X_2,X_3,X_4\}$. Besides the linear relation  
$x_1+x_2+x_3+x_4=\cte_1$, we get another independent relation: 
$\kappa_4\kappa_1x_2-\kappa_4\kappa_2x_3+\kappa_1(\kappa_3+\kappa_4)x_4=\cte_2$.

\subsection*{The associated digraphs}

\medskip

Consider a directed graph  $G=(\Ve,\Edg, \Ka)$ with a partition of the set of species 
which defines a MESSI structure in the network.
We associate to $G$ three other digraphs, denoted by $G_1, G_2$, $G_E$.


\smallskip

\begin{definition}\label{def:G1}
Given a chemical reaction network with directed graph $G=(\Ve,\Edg,\Ka)$, 
together with a partition of the set of species $\Sp$ which defines
a MESSI structure in the network with $p$ intermediate species and $n$ core species as in~\eqref{eq:S}, 
we associate a digraph $G_1=(\Ve_1,\Edg_1)$ with a set of $n$ {species} consisting of the core species in $G$
and with the inherited partition:
\begin{equation}\label{eq:G1}
 \Sp_1= \Sp^{(1)} \bigsqcup \Sp^{(2)} \bigsqcup \dots \bigsqcup \Sp^{(m)} = \Sp \setminus \Sp^{(0)}.
\end{equation}
The vertex set $\Ve_1$ consists of all the core complexes $y_{ij}$
and the edge set  is equal to $\Edg_1=\{y_{ij}\to y_{\ell m}: y_{ij},y_{\ell m} \in \Ve_1  \text{ and }  
y_{ij}\uri y_{\ell m}{\text{ in } G}\}$.
\end{definition}

Note that $G_1$ might have loops.
It is easy to check that partition~\eqref{eq:G1} defines a MESSI structure on $G_1$ for any choice
of positive labels in $\R_{>0}^{\# \Edg_1}$.  

We  now define a chemical reaction network on $G_1$ by decorating the edges $\Edg_1$ with labels $\Katau(\Ka)$, which
are rational functions of the original rate constants $\Ka$, 
following~\cite[Theorem~3.1]{fw13}.

\begin{definition}\label{def:tau}
The map $\Katau:\R_{>0}^{\# \Edg}\to \R_{>0}^{\# \Edg_1}$ is defined as follows. For each $X_i+X_j\uri X_\ell+X_m$ in $G$  
the reaction constant $\tau$ in $G_1$ which gives the label $X_i+X_j\overset{\tau}{\longrightarrow} X_\ell+X_m$ has the form
\begin{equation} \label{eq:tau}
\tau=\kappa+\overset{p}{\underset{k=1}{\sum}}\kappa_k\mu_k,
\end{equation}
where $\kappa\ge 0$ is positive when $X_i+X_j\overset{\kappa}{\longrightarrow} X_\ell+X_m$ in $G$ 
(and $\kappa=0$ otherwise), and $\kappa_k\ge 0$ is positive if $U_k\overset{\kappa_k}{\longrightarrow} X_\ell+X_m$ 
and $X_i+X_j\uri U_k$ in $G$ (and $\kappa_k=0$ otherwise).  The explicit expression of the coefficients 
$\mu_k$ is given in display (15) in the proof of Theorem~3.1 in the electronic supplementary material (ESM) 
of~\cite{fw13}; we will describe them for particular cases of interest to us in Section~\ref{sec:param}.
\end{definition}

It is straightforward to see that  $\Katau$ defines a rational map (that is,$\Q(\tau) \subset \Q(\Ka)$). 
The main property of this assignment is the following.

\begin{remark}\label{rem:tauq}
When we label the edges in $G_1$ with 
the real constants $\Katau(\Ka)\in \R_{>0}^{\# \Edg_1}$, the steady states of the mass-action chemical 
reaction systems defined by $G$ and $G_1$ are in one-to-one correspondence. We refer the reader 
to the proof of Theorem~3 in the ESM of \cite{fw13} and to the more recent article~\cite{MFW16}.
\end{remark}

We now introduce a new associated labeled digraph $G_2$.  

\begin{definition}\label{def:G2}
Consider a chemical reaction network with directed graph $G=(\Ve,\Edg,\Ka)$, 
together with a partition of the set of species $\Sp$ which defines
a MESSI structure in the network, and its associated labeled digraph $G_1=(\Ve_1,\Edg_1, \Katau)$ 
from Definition~\ref{def:G1}.
We first define a labeled multidigraph
where we ``hide'' the concentrations of some of the species in the labels.
 The species set $\Ve_2$ of $G_2=(\Ve_2,\Edg_2, \Katau_x)$ is again equal to
the set of core species $\Sp_1$, with the induced partition. 

The edge set $\Edg_2$ is defined as follows.
We keep all monomolecular reactions $X_i\to X_j$ in $\Edg_1$ and
for each reaction $X_i+X_\ell \overset{\tau}{\longrightarrow} X_j+X_m$ in $\Edg_1$,
with $X_i,X_j \in \Sp^{(\alpha)}$, $X_\ell,X_m \in \Sp^{(\beta)}$,
we consider two reactions $X_i \overset{\tau x_\ell}{\longrightarrow} X_j$
and $X_\ell \overset{\tau x_i}{\longrightarrow} X_m$.
We obtain in principle a multidigraph $MG_2$ that might contain loops or parallel edges 
between any pair of nodes (i.e., directed edges with the same source and target nodes).
We define the digraph $G_2$ by collapsing into one edge all parallel edges in $MG_2$, 
and we define the labels $\Katau_x$ of the edges in $\Edg_2$ as the sum of the labels 
of the corresponding collapsed edges in $MG_2$.

We will moreover denote by $G_2^\circ$ the digraph obtained from the deletion of loops 
and isolated nodes of $G_2$.
\end{definition}

By rules $({\mathcal R}_1)$, $({\mathcal R}_2)$ and $({\mathcal R}_3)$, $G_2$ is a 
\emph{linear} graph (its vertices are labeled by a single species).  The labels on 
the edges of $MG_2$ (and of $G_2$) depend on the rate constants but might also depend 
on the concentrations $x_1, \dots, x_n$. 

\begin{example}[Examples~\ref{ex:ES} and~\ref{ex:S}, continued]
The graphs $G_1$ and $G_2^\circ$ associated to the networks in Examples~\ref{ex:ES} and~\ref{ex:S} are depicted in 
Figure~\ref{fig:Gs}.

\begin{figure}
\centering
 \begin{tabular}{|>{\centering}m{1.3in} >{\centering}m{0.1in}  >{\centering}m{1.3in} >{\centering\arraybackslash}m{1.5in}|}
 \hline
 \multicolumn{1}{|l}{$G_1$:} &  & \multicolumn{1}{l}{$G_2^\circ$:} & \multicolumn{1}{l|}{$G_E$:}\\
 $\begin{array}{l}
  S_0+E \overset{\tau_1}{\rightarrow} S_1+E \\
  S_1+F \overset{\tau_2}{\rightarrow} S_0+F \\
  P_0+S_1 \overset{\tau_3}{\rightarrow} P_1+S_1\\
  P_1+F \overset{\tau_4}{\rightarrow} P_0+F
\end{array}$ & 
\text{\Large \textcolor{blue}{$\Rightarrow$}} & 
$\begin{array}{l}
S_0\underset{\tau_2f}{\overset{\tau_1e}{\rightleftarrows}} S_1\\
P_0\underset{\tau_4f}{\overset{\tau_3s_1}{\rightleftarrows}} P_1
\end{array}$ & 
\begin{tikzpicture}[ampersand replacement=\&] 
\matrix (m) [matrix of math nodes, row sep=1.5em, column sep=1em, text height=1.5ex, text depth=0.25ex]
{ \Sp^{(3)}\&  \Sp^{(1)} \& \Sp^{(2)} \\
\Sp^{(4)} \&  \& \\};
\draw[->]($(m-1-1)+(0.4,0)$) to node[below] (x) {} ($(m-1-2)+(-0.4,0)$);
\draw[->]($(m-2-1)+(0.4,0)$) to node[below] (x) {} ($(m-1-2)+(-0.35,-0.15)$);
\draw[->]($(m-1-2)+(0.4,0)$) to node[below] (x) {} ($(m-1-3)+(-0.4,0)$);
\draw[->]($(m-2-1)+(0.4,0)$) to node[below] (x) {} ($(m-1-3)+(-0.4,-0.1)$);
\end{tikzpicture}\\ 
& & & \\
\multicolumn{1}{|l}{$G_1$:} &  & \multicolumn{1}{l}{$G_2^\circ$:} & \multicolumn{1}{l|}{$G_E$:}\\
\begin{tikzpicture}
 \node[]at(0,0)(a){$\begin{array}{l}
  X \underset{\tau_2}{\overset{\tau_1}{\rightleftarrows}} XT
  \overset{\tau_3}{\rightarrow} X_p\\
  X_p+Y \overset{\tau_4}{\rightarrow} X+Y_p \\
  XT+Y_p \overset{\tau_5}{\rightarrow} XT+Y
  \end{array}$};
\end{tikzpicture} & 
\text{\Large \textcolor{blue}{$\Rightarrow$}} & 
\begin{tikzpicture}[ampersand replacement=\&] 
  \matrix (m) [matrix of math nodes, row sep=2.5em, column sep=0.1em, text height=1.5ex, text depth=0.25ex]
    {X \& \overset{\tau_1}{\underset{\tau_2}{\rightleftarrows}} \& XT \& \overset{\tau_3}{\rightarrow}  \& X_p \\};
  \draw[->](m-1-5) to[in=-45,out=-135] node[below] (x) {\footnotesize $\tau_4y$} (m-1-1);
  \node[]at(0,-2)(a){$Y \underset{\tau_5xt}{\overset{\tau_4x_p}{\rightleftarrows}} Y_p$};
\end{tikzpicture} & 
$\Sp^{(1)}\rightleftarrows \Sp^{(2)}$\\
\hline
\end{tabular}\\

\caption{The graphs $G_1$, $G_2^\circ$ and $G_E$ for the running examples. 
The corresponding sets $\Sp^{(\alpha)}$ can be found 
in Example~\ref{ex:ESS}}\label{fig:Gs}

\end{figure}
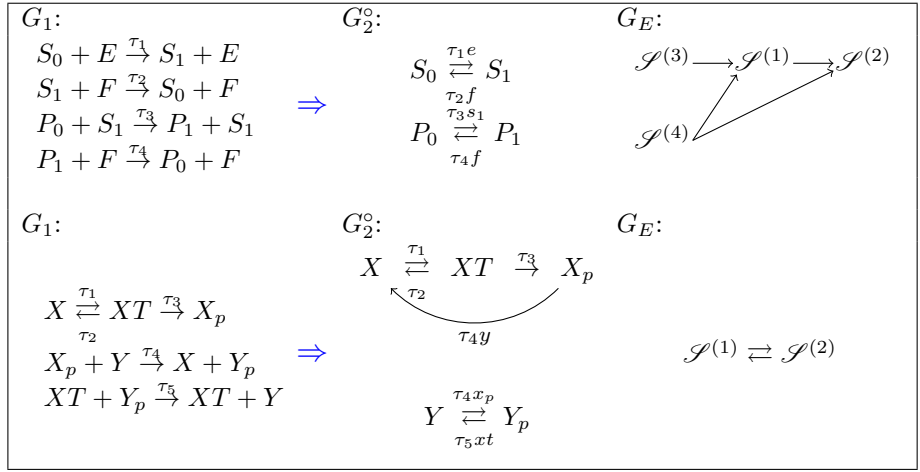
\end{example}

\smallskip

\begin{remark}\label{rem:G2}
We  get the following important fact from the definition of the associated digraphs and networks
for any MESSI network with digraph $G$: 
the networks of the associated digraphs $G_1$ and $G_2$  \emph{determine the same polynomial equations}. 
They moreover define, together with the corresponding equations of the intermediate species, 
the steady states of $G$. We have already observed in Remark~\ref{rem:tauq} that 
the steady states of $G_1$ and $G$ are in one-to-one correspondence.
Indeed, if we consider $G_2$ in a mass-action fashion, we can see that the same
terms are added and substracted, obtaining the same equations associated to $G_1$.
However, we cannot recover the dynamical properties of $G_1$ (nor $G$)
from $G_2$ since we admit species (concentrations) as both vertices and edge labels. 
\end{remark}

Note that for each $\alpha\ge1$, if one species of $\Sp^{(\alpha)}$ appears on a vertex
of $G_2$, by $({\mathcal R}_2)$ and $({\mathcal R}_3)$ and the construction of $G_2$, all the species in the vertices of
the corresponding connected component of $G_2$ belong to the \emph{same}
 subset $\Sp^{(\alpha)}$ in the original partition~\eqref{eq:S}. 
In fact, the same partition~\eqref{eq:G1} defines a MESSI structure on $G_2$. Moreover, we have the following.

\begin{lemma}\label{lem:min}
The partition of the set of species $\Sp$ of $G$ in~\eqref{eq:S}
is minimal in the poset of partitions defining a MESSI structure on the network 
if and only if the set of intermediate species is maximal, the connected components 
of $G_2$ are in bijection with the subsets $\Sp^{(\alpha)}$, and the set of nodes 
of the corresponding component equals $\Sp^{(\alpha)}$.
Thus, by considering the connected components in $G_2$ we can refine any partition 
of the species set $\Sp$ to a minimal one defining a MESSI structure on $G$.
\end{lemma}

We finally define the associated digraph $G_E$.

\begin{definition}\label{def:GE}
Consider a MESSI network with  directed graph $G$, together with a \emph{minimal} partition of the set of species as in~\eqref{eq:S}. Let $G_2$
 and  $G_2^0$ be as in Definition~\ref{def:G2}.
We define a new digraph $\GE= (\Ve_E, \Edg_E)$. The set of vertices equals $\Ve_E= \{\Sp^{(\alpha)}, \, \alpha \ge 1\}$.
The pair $(\Sp^{(\alpha)},\Sp^{(\beta)})$ lies in $\Edg_E$ when there is a species
in $\Sp^{(\alpha)}$ in a {label} of an edge in $G_2^0$ between (different) species of $\Sp^{(\beta)}$.
\end{definition}

Example~\ref{ex:ESS} below shows the corresponding digraphs $G_E$ for our two 
running examples.

\subsection*{Persistence}\label{sssec:boundary}

\medskip

As MESSI systems are conservative by Theorem~\ref{th:conservations}, we know by Theorem~2 in~\cite{ADLS07} 
that a MESSI system is persistent when there are no \emph{relevant boundary steady states}. This means that 
there are no steady states in the intersection of the boundary $\partial(\R^s_{\ge 0})$ of the nonnegative 
orthant with a stoichiometric compatibility class through a point in $\R^s_{>0}$. Persistence means that any 
trajectory starting from a point with positive coordinates stays at a positive distance from any point in the 
boundary. 

Note that a necessary condition for system~\eqref{CRN} to have a positive steady  state is the existence of 
a positive relation among the vectors $y'-y$, that is, a positive vector $\lambda$ such that 
$\sum_{y \to y'}\lambda_{yy'} (y'-y)=0$. If this is satisfied, we will say that the system is \emph{consistent}.

We give in Theorem~\ref{th:bss} combinatorial conditions which ensure the persistence and consistency of MESSI systems.
This result rules out relevant boundary steady states in many enzymatic examples--for instance, those in~\cite{ADLS07}.

Recall that a digraph is weakly reversible if any connected component is strongly connected, that is,
when for any pair of nodes in the same connected component there is
a directed path joining them. 
We have the following persistence result.

\begin{theorem}\label{th:bss}
Let $G$ be the underlying digraph of a MESSI system. Assume that the associated digraph $G_2$  
is weakly reversible and the associated digraph $\GE$ has no directed cycles. 
Then $G$ has no relevant boundary steady states and so the system is persistent. 
Moreover, the system is consistent.
\end{theorem}

\begin{remark}\label{rem:m}
The absence of directed cycles in $\GE$ precludes the existence of swaps. 
On the other side, note that if $G_2$ is weakly reversible, then the stoichiometric 
and the kinetic subspaces coincide by Proposition~\ref{prop:kinstoi}. 
\end{remark}

\begin{example}[Examples~\ref{ex:ES} and~\ref{ex:S}, continued]\label{ex:ESS}
The MESSI network in Example~\ref{ex:ES} from Figure~\ref{fig:sadi} (A)  (with partition
$\Sp^{(1)}=\{S_0,S_1\}$,  $\Sp^{(2)}=\{P_0,P_1\}$,   $\Sp^{(3)}=\{E\}$,  $\Sp^{(4)}=\{F\}$) 
is persistent since there are no directed cycles in $G_E$ 
(depicted at the upper right in Figure~\ref{fig:Gs}).
However, this is not the case in Example~\ref{ex:S} from Figure~\ref{fig:4}(D); 
$x_p=X_{tot}, y_p=Y_{tot}, x=xt=x_py=xty_p=y=0$
is a boundary steady state in the stoichiometric compatibility class defined by $X_{tot}, Y_{tot}$. Recall
that we are considering the (minimal) partition $\Sp^{(1)}=\{X,XT,X_p\}$,   $\Sp^{(2)}=\{Y,Y_p\}$. 
The associated graph $\GE$ has a cycle (depicted at the lower right in Figure~\ref{fig:Gs}).
\end{example}


\section{Parametrizing the steady states} \label{sec:param}%

A wide class of MESSI systems admits a rational parametrization. 
As we recalled in Remark~\ref{rem:G2}, it is shown  in~\cite{fw13} that 
the values of the intermediate species at steady state can be rationally 
written in terms of the core species in an algorithmic way. 
The following result (with the same assumptions 
as Theorem~\ref{th:bss}) extends Theorem~4 in \cite{TG09}.

\begin{theorem}\label{th:par}
Let $G$ be the underlying digraph of a MESSI system.
Assume that the associated digraph $G_2$  is weakly reversible and the associated digraph $\GE$ has no directed cycles. 
Then, $V_f \cap \R_{>0}^s$ admits a rational parametrization, which can be algorithmically computed. More explicitly, 
it is possible to define levels for the subsets $\Sp^{(\alpha)}, \alpha \ge 1,$ according to indegree. 
Then, given any choice of one index $i_\alpha$ in each $\Sp^{(\alpha)}$, the concentration of any core species $x_i$ in 
a subset  $\Sp^{(\beta)}$ can be rationally expressed in an effective way in terms of $x_{i_\beta}$ and the variables 
$x_{i_\alpha}$ for which the indegree of $\Sp^{(\alpha)}$ is strictly smaller than the indegree of $\Sp^{(\beta)}$.

Moreover, if the partition is minimal with $m$ subsets of core species,  the dimension of $V_f \cap \R_{>0}^s$ 
equals $m$ and $m=\dim (\stoich^\bot)$.
\end{theorem}

Recall that a binomial is a polynomial with two terms and that a Laurent monomial is a monomial 
with integer exponents, which can be negative.
\begin{definition}\label{def:toric}
 A \emph{toric MESSI system} is a MESSI system whose positive steady states $V_f \cap \R_{>0}^s$ 
 can be described  with \emph{binomials}.
\end{definition}

It is well known that the real positive points of a nonempty algebraic variety described by binomials 
can always be parametrized by Laurent monomials.  This implies that if the MESSI system is toric, 
there exists a rational parametrization even if $\GE$ has directed cycles, as long as the system 
is consistent.

We now show that many common MESSI systems are toric in an explicit way coming
from the structure of the network, which we call s-toric. 

In order to define s-toric MESSI systems, we need to use some concepts from graph theory. 
A spanning tree of a digraph is a subgraph that contains all the vertices and is connected and acyclic
as an undirected graph. An $i$-tree of a graph is a spanning tree where the $i$th vertex
is its unique sink (equivalently, it is the only vertex of the tree with no edges leaving from it).
For an $i$-tree $T$, call $c^{T}$ the product of the labels of all the edges of $T$.
For the associated graph $G_2$ of a MESSI network $G$, the products $c^T$ are
monomials depending in principle on both the  rate constants $\tau$ and the $x$-variables.

\begin{definition}\label{def:storic}
 A \emph{structurally toric}, or \emph{s-toric MESSI system}, is a MESSI system whose digraph $G$ 
 satisfies the following conditions:
 \begin{itemize}
 \item[$(\cond')$] Condition $(\cond)$ holds, and moreover, for every intermediate complex $y_k$ there exists 
 a \emph{unique} core complex $y_{ij}$  such that $y_{ij}\uri y_k$ in $G$. 
  
  \item[$(\cond'')$] The associated multidigraph $MG_2$ does not have parallel edges, and the
 digraph $G_2$ is weakly reversible.
 
  \item[$(\cond''')$] For each $i \in \{1, \dots, n\}$ and any choice of $i$-trees $T,T'$ of $G_2^\circ$,
 the quotient $c^T/c^{T'}$ only depends on the rate constants $\tau$.
 \end{itemize}
\end{definition}

Examples of networks satisfying condition $(\cond''')$  are the phosphorylation cascades, as there is a unique
$i$-tree for each $i$. Our second running Example~\ref{ex:S} also has this property (see 
Example~\ref{ex:Sregulated}). 
Moreover, phosphorylation cascades, the multisite sequential distributive phosphorylation system,
the multisite processive phosphorylation system, and the bacterial EnvZ/OmpR network depicted in Figure~\ref{fig:4}
are s-toric MESSI systems.

\begin{example}[Running Example~\ref{ex:S}, continued]\label{ex:Sregulated}
For the system in Example~\ref{ex:S}, 
the graph $G_2^\circ$ is:
\begin{center}
 \begin{tabular}{ll}
   \multirow{2}{*}{\begin{tikzpicture}[ampersand replacement=\&] 
   \matrix (m) [matrix of math nodes, row sep=2.5em, column sep=0.1em, text height=1.5ex, text depth=0.25ex]
      {X \& \overset{\tau_1}{\underset{\tau_2}{\rightleftarrows}} \& XT \& \overset{\tau_3}{\rightarrow}  \& X_p \\};
   \draw[->](m-1-5) to[in=-45,out=-135] node[below] (x) {\footnotesize $\tau_4y$} (m-1-1);
   \end{tikzpicture}} &
   $Y \underset{\tau_5x\scriptscriptstyle{T}}{\overset{\tau_4x_p}{\rightleftarrows}} Y_p.$\\
   & \\
 \end{tabular}
\end{center}
In this case, there are two $X$-trees:
\begin{center}
  \begin{tabular}{ll}
    \begin{tikzpicture}[ampersand replacement=\&] 
     \matrix (m) [matrix of math nodes, row sep=2.5em, column sep=0.1em, text height=1.5ex, text depth=0.25ex]
      {T_1:  \& X \& \underset{\tau_2}{\leftarrow} \& XT \&   \& X_p \\};
     \draw[->](m-1-6) to[in=-45,out=-135] node[below] (x) {\footnotesize $\tau_4y$} (m-1-2);
     \end{tikzpicture} &
    \begin{tikzpicture}[ampersand replacement=\&] 
     \matrix (m) [matrix of math nodes, row sep=2.5em, column sep=0.1em, text height=1.5ex, text depth=0.25ex]
      {T_2:  \& X \&  \& XT \& \overset{\tau_3}{\rightarrow}  \& X_p. \\};
     \draw[->](m-1-6) to[in=-45,out=-135] node[below] (x) {\footnotesize $\tau_4y$} (m-1-2);
     \end{tikzpicture}
 \end{tabular}
\end{center}
\noindent However,  $c^{T_1}=\tau_2\tau_4y$,  $c^{T_2}=\tau_3\tau_4y$, and  $c^{T_1}/c^{T_2}=\tau_2/\tau_3$, 
which only depends on the rate constants $\tau_i$.
For the other vertices, the corresponding tree is unique, and therefore this MESSI network is s-toric.
\end{example}

We now clarify the meaning of condition $(\cond')$.

\begin{example}\label{ex:condCprime}
Network (A) on the left of Figure~\ref{fig:c'} satisfies condition~$(\cond')$, 
while network (B) on the right does not since both core complexes 
$X_1$ and $X_2$ react via intermediates to the intermediate complex $U_2$.
{\small
\begin{figure}[ht]
\begin{center}
\begin{tabular}{ccccc}
 & \multirow{3}{*}{
\begin{tikzpicture}
\matrix (m) [matrix of math nodes, row sep=0.3em,ampersand replacement=\&,column sep=1em]
{
\& U_1\& X_3 \\
X_1\& \& X_2 \\
\& U_2 \& \\
};
\path[-stealth]
(m-2-1) edge (m-1-2)
(m-2-1) edge (m-3-2)
(m-1-2) edge (m-2-3)
(m-1-2) edge (m-1-3)
(m-3-2) edge (m-2-3);
\end{tikzpicture}
}    & & & \\
{\footnotesize (A)}& & & {\footnotesize (B)} & $X_1\rightleftarrows U_1 \rightarrow U_2 \rightleftarrows X_2.$\\
 & & & & \\
 & & & &
\end{tabular}
\end{center}
\caption{Validity of condition $(\cond')$}\label{fig:c'}
\end{figure}
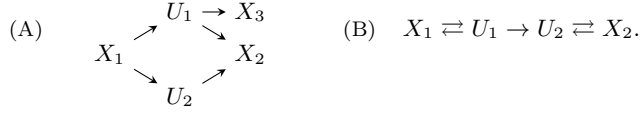}
\end{example}

We will use the following notation.

\begin{notation}\label{not:phi}
Given an intermediate complex $y_k$ of an s-toric MESSI system, denote by $y_{ij}$ the unique 
core complex reacting through intermediates to  $y_k$ and denote by $\bx^{\varphi(k)}$ the monomial
\begin{equation}\label{eq:phi}
 \bx^{\varphi(k)}=\left\lbrace \begin{array}{l c l}
                           x_{i}x_{j} & \text{if} & y_{ij} =X_{i}+ X_{j}\\
                           x_{i} & \text{if} &  j=0 \text{ and } \, y_{ij}= X_{i}.
                          \end{array}\right.
\end{equation}
\end{notation}

As we recalled in Remark~\ref{rem:tauq},  the rational map 
$\Katau:\R_{>0}^{\# \Edg}\to \R_{>0}^{\# \Edg_1}$ in Definition~\ref{def:tau} verifies 
that the steady states of the mass-action chemical reaction systems defined by $G$ with rate 
constants $\Ka$ and $G_1$ with rate constants $\Katau(\Ka)$ are in one-to-one correspondence 
via the projection $\pi({u},\bx)=\bx$. 
We now give conditions for the inverse of this projection to be a monomial map in the 
concentrations of the core species.

\begin{proposition}\label{prop:int}
Given a MESSI network $G$ that satisfies condition ($\cond'$) in Definition~\ref{def:storic},  
there are (explicit) rational functions $\mu_k\in\Q(\Ka), 1\le k\le p$,  such that for any 
steady state $\bx \in \R^n_{>0}$ of the associated MESSI network $G_1$, the steady 
state $\pi^{-1}(\bx) = ({u}(\bx), \bx)$ of $G$ is given by the monomial map:
\begin{equation}\label{eq:monom}
 u_k(\bx) \, = \, \mu_k \, \bx^{\varphi(k)}, \quad k=1,\dots, p.
\end{equation}
The rational functions $\mu_k$ are in simple cases the usual Michaelis--Menten constants associated 
with the original rate constants $\Ka$.
\end{proposition}


It holds that an s-toric MESSI system  is toric and, moreover, its 
positive steady states can be described by explicit binomials.

\begin{theorem}\label{th:toric_toric}
 Any  s-toric MESSI system is toric. Moreover, we can choose $s-m'$ explicit binomials with coefficients in
$\Q(\Ka)$ which describe the positive steady states, where $m'$ is the number of connected components of $G_2$.
 
In particular, given a MESSI network $G$ with a partition of the set of species as in~\eqref{eq:S},
assume that for each $\alpha\ge1$ and $X_i \neq X_j\in\Sp^{(\alpha)}$ in the same
connected component of $G_2$ there exists a unique simple path $P_{ji}$ in $G_2^\circ$ from 
$X_j$ to $X_i$.\footnote{A \emph{simple path} is a path that visits each vertex exactly once.} 
Then, the associated dynamical system is s-toric and there exist explicit $\mu_k$ and $\eta_{ij}$ 
in $\Q(\Ka)$ such the $s-m'$ binomials describing the positive steady states can be chosen from the following:
 
\begin{minipage}{0.7\textwidth}
 \begin{align}
\label{eq:binom_w_3} & u_k-\mu_k \bx^{\varphi(k)}=0\\
\nonumber & \text{ for each intermediate }U_k  \, (1\leq k\leq p),\\
\label{eq:binom_w_4} & x_hx_i-\eta_{ij} x_m x_j=0 \\
\nonumber & \text{ if } X_i\overset{\tau x_h}{\longrightarrow}X_j
  \text{ is in }G_2^\circ \text{ and } X_j\overset{\tau'x_m}{\longrightarrow}X_\ell \text{ is in } P_{ji}.
\end{align}
\end{minipage}
\begin{minipage}{0.25\textwidth}
{\footnotesize
\begin{tikzpicture}
  \matrix (m) [matrix of math nodes, row sep=0.5em,ampersand replacement=\&,
    column sep=0.5em]{
    X_i \& \overset{\tau x_h}{\longrightarrow} \& X_j \\
    \& \& \\
     \& \& X_\ell\\};
     \draw[->,dashed](m-3-3) to[in=-55,out=190] node[below] (x) {} (m-1-1);
     \draw[->](m-1-3) to node[right] (x) {\tiny $\tau' x_m$} (m-3-3);
 \end{tikzpicture}
 }
\end{minipage}
\end{theorem}

\smallskip

\begin{example}[Running Example~\ref{ex:ES}, continued]\label{ex:regulated}
Recall that the graph $G_2^\circ$ for the cascade in Example~\ref{ex:ES} is
\[S_0\overset{\tau_1e}{\underset{\tau_2f}{\rightleftarrows}}S_1
 \quad P_0\overset{\tau_3s_1}{\underset{\tau_4f}{\rightleftarrows}}P_1,\]
and the graph $G_2$ has two extra connected components, corresponding to the isolated nodes $E$ and $F$.
Clearly, for each vertex in $G_2^\circ$ there is only one simple directed path from 
the other vertex in the same connected component.
For example, the only $S_1$-tree, $T$, is $S_0\overset{\tau_1e}{\rightarrow}S_1$ and
$c^T=\tau_1e$.

We denote the concentration of the intermediate species $es_0, fs_1,s_1p_0,fp_1$ by $u_1, u_2, u_3, u_4$, 
respec\-tively. The corresponding rational functions $\mu_1, \dots, \mu_4$ in
the statement of Proposition~\ref{prop:int} equal
 \[\mu_1=\frac{\kappa_1}{\kappa_2+\kappa_3}, \mu_2=\frac{\kappa_4}{\kappa_5+\kappa_6}, 
 \mu_3=\frac{\kappa_7}{\kappa_8+\kappa_9},
 \mu_4=\frac{\kappa_{10}}{\kappa_{11}+\kappa_{12}}.\]
We further denote
$\eta_1=\frac{\tau_2}{\tau_1}, \eta_2=\frac{\tau_4}{\tau_3}$.
According to Theorem~\ref{th:toric_toric}, the following $6=10-4$ binomials describe the positive 
steady states of the associated MESSI system: 
$$u_1-\mu_1 e.s_0 \, = \, u_2-\mu_2 f.s_1 \, = \,
u_3-\mu_3 s_1.p_0 \, = \, u_4-\mu_4 f.p_1 \, = \,  e.s_0-\eta_1 f.s_1 \, = \,  s_1.p_0-\eta_2 f.p_1 =0.$$
The first four binomials correspond to~\eqref{eq:binom_w_3}, and the last two occur 
in~\eqref{eq:binom_w_4}.
\end{example}

\section{Toric MESSI systems and Multistationarity}\label{sec:multistationarity}

We present in this section a necessary and sufficient criterion to decide whether a system is 
multistationary, which holds for toric MESSI systems (see Definitions~\ref{def:toric} 
and~\ref{def:storic}).
Again, the assumptions we make seem to be very restrictive. Nevertheless, it can be easily seen 
that all standard phosphorylation cascades, multisite sequential phosphorylation networks and 
many two component bacterial networks are of this form, so there is a wide range of applications.  
This is summarized in Theorems~\ref{th:monostationarity} and~\ref{th:multistationarity}. 
We implemented this result by means of Algorithm~\ref{algo}, which certifies mono- or multistationarity, 
and in this last case provides different choices of rate constants for which multistationarity occurs.

\subsection*{Necessary and sufficient conditions}

\medskip

Theorem~\ref{th:monostationarity} below gives a necessary and sufficient criterion to detect the
capacity for multistationarity of a toric MESSI system. It is deduced from results in
\cite{mfrcsd13} and \cite{PM12}. Then, we give in Proposition~\ref{prop:rankM} checkable 
conditions that ensure the validity of the hypotheses of Theorem~\ref{th:monostationarity}. 
When the system is not monostationary, we finally show in Theorem~\ref{th:multistationarity} 
how to choose rate constants for which the system shows multistationarity (see also 
\cite{CoFlRa08,Fein95DefOne}).

\begin{notation}\label{not:BTM}
Let $G$ be a MESSI network.
Assume the positive steady states of the associated dynamical system are described
by binomials $x^{v'}-\eta x^{v}$, with $v,v'\in \Z_{\ge 0}^n$. We call $\bino$ the subspace of $\R^n$ 
generated by all the vectors $v'-v$.
Choose any matrix $\binoM$ whose columns form a basis
of $\bino$. For a positive vector $x$ write  $(x^\binoM)_j=x^{\binoM_j}$,
where $\binoM_j$ denotes the $j$th column of $\binoM$. 
Then, there exists a constant vector $\boldsymbol{\eta}$ such that $x$ is a positive steady state 
of the associated system if and only if $x^\binoM=\boldsymbol{\eta}$. 
Considering the orthogonal complement of $\bino$ in $\R^s$, we construct another matrix $\binoMp$ 
whose rows form a basis of the orthogonal subspace $\bino^\bot$. We can choose both $\binoM$ and 
$\binoMp$ with integer entries.
We consider also a matrix $\stoichM$ whose columns form a basis
of  the stoichiometric subspace $\stoich$. Again, we construct a matrix $\stoichMp$ whose
rows form a basis of the orthogonal complement $\stoich^\bot$. Thus, when
the stoichiometric and the kinetic spaces coincide, the row vectors of $\stoichMp$ are
the coefficients of a basis of linear conservation relations. 
For any natural number $s$ we denote $[s]=\{1, \dots, s\}$. Given a matrix
$A\in \mathbb{R}^{d\times s}$ with $s\ge d$ and a subset $J\subseteq [s]$, we
denote  by $A_J$ the submatrix of $A$ with column indices in $J$.  We furthermore denote $J^c$
the complement of $J$ in $[s]$ and $\nu(J)=\sum_{j \in J} j$. 
An orthant  $\mathcal{O} \subset \R^s$ is defined by the signs of the coordinates of its
points and it will be identified with a vector in  $\{-1,0,1\}^s$.
\end{notation}

\begin{definition}\label{def:signs}
Given matrices $\stoichMp$ and $\binoMp$ as above, with $d=\rank(\stoichMp)=\rank(\binoMp)$, we define the 
following sets of signed products:
 \[\begin{array}{lcl}
 \Sigma & = & \{\sign(\det(\stoichM^t_{I})\det(\binoM^t_{I})): \, I\subseteq [s], \, \#I=s-d\},\\
 \Sigma^\bot & = &\{\sign((-1)^{\nu(J)}\det(\stoichMp_{J})\det(\binoM^t_{J^c})): \, J\subseteq [s], \, \#J=d\},\\
 \Sigma_\bot & = &\{\sign((-1)^{\nu(J)}\det(\stoichM^t_{J^c})\det(\binoMp_{J})): \, J\subseteq [s], \, \#J=d\},\\
 \Sigma^\bot_\bot & = & \{\sign(\det(\stoichMp_{J})\det(\binoMp_{J})): \, J\subseteq [s], \, \#J=d\}.
 \end{array}\]
We say that a set $\sigma \not= \{0\}$ of signs is \emph{mixed} if $\{-,+\}\subset \sigma$ and \emph{unmixed} 
otherwise.
\end{definition}
The following lemma is a consequence of Lemma~2.10 in \cite{mfrcsd13} (and the references therein).
\begin{lemma}\label{lem:signs}
With the notation of Definition~\ref{def:signs}, if any of the four signs sets $\Sigma, \Sigma^\bot,  
\Sigma_\bot,\Sigma^\bot_\bot$ is different from $\{0\}$, the four of them are, and if so, if any of
the four is mixed, all of them are mixed.
\end{lemma}

The following theorem gives a necessary and sufficient criterion to determine if the toric 
MESSI system is monostationary, based on \cite{mfrcsd13} and \cite{PM12}.

\begin{theorem}\label{th:monostationarity}
Let $G$ be a toric MESSI network with matrices $\stoichM$ and $\binoM$ as above, 
which verifies that $\rank(\stoichM)=rank(\binoM)=d$ and the signs sets $\Sigma, \Sigma^\bot,  
\Sigma_\bot,\Sigma^\bot_\bot$ are different from $\{0\}$.
Then, the following statements are equivalent:

\begin{enumerate}
 \item \label{itm:mono} The associated MESSI system is monostationary.
 \item \label{itm:unmixed} The signs sets $\Sigma, \Sigma^\bot,  
 \Sigma_\bot,\Sigma^\bot_\bot$ are unmixed.
 \item \label{itm:orthant} For all orthants $\mathcal{O}\in\{-1,0,1\}^s,
 \mathcal{O}\neq\mathbf{0}$, either $\stoich\cap \mathcal{O}= \emptyset$
 or $\bino^\bot\cap\mathcal{O}= \emptyset$.
\end{enumerate}
\end{theorem}

\begin{example}[Example~\ref{ex:ES}, continued] 
Consider the two phosphorylation cascades in Figure~\ref{fig:sadi}.
Both cascades differ in the phosphatases: the cascade in Figure~\ref{fig:sadi} 
(B) has different phosphatases for each layer, while the cascade  (A) does not. 
The set $\Sigma$ corresponding to the cascade in (B) is unmixed, which according 
to Theorem~\ref{th:monostationarity} implies that the system is monostationary. 
In contrast, the set $\Sigma$ for the cascade in (A) is mixed, and the system has 
the capacity for multistationarity. For instance, if we consider $J$ the set of
indices corresponding to $S_0, P_0, ES_0$, and $FP_1$, and $\tilde{J}$ the set of 
indices corresponding to $S_0, P_1, ES_0$, and $FP_1$ (where $4=\rank(\stoichMp)=\rank(\binoMp)$),
$\sign(\det(\stoichMp_{J})\det(\binoMp_{J}))\neq \sign(\det(\stoichMp_{\tilde{J}})\det(\binoMp_{\tilde{J}}))$,
and they are both nonzero.

If we add the reactions $P_1+D\rightleftarrows P_1D$, which represent a drug interacting
with the phosphorylated form $P_1$, we can check that this new system remains multistationary for the cascade (A).
The new matrices $\hat{\stoichM}$ and $\hat{\binoM}$ can be obtained
in the following way:

\begin{minipage}{0.45\textwidth}
{\small
 \begin{tabular}{ccccccccc}
 & &  & $\phantom{00}$  & & $\, {\hspace{0.1cm}\scriptstyle P_1}$ & ${\scriptstyle D}$ & ${\scriptstyle P_1D}$ & \\
 $\hat{\stoichM}^t=$& \multicolumn{8}{l}{$\left(\begin{array}{cccccc}
                      & &  & \multicolumn{1}{c|}{*} & 0 & \phantom{-}0\\
                      & \stoichM^t & & \multicolumn{1}{c|}{\vdots} & \vdots & \phantom{-}\vdots \\
                      & & & \multicolumn{1}{c|}{*} & 0 & \phantom{-}0\\\cline{1-4}
                      0 & \dots &  0 & 1 & 1 & -1
                    \end{array}\right)$,}
\end{tabular}
}
\end{minipage}
\begin{minipage}{0.45\textwidth}
{\small
 \begin{tabular}{ccccccccc}
 & &  & $\phantom{00}$  & & $\, {\hspace{0.1cm}\scriptstyle P_1}$ & ${\scriptstyle D}$ & ${\scriptstyle P_1D}$ & \\
 $\hat{\binoM}^t=$& \multicolumn{8}{l}{$\left(\begin{array}{cccccc}
                      & &  & \multicolumn{1}{c|}{*} & 0 & \phantom{-}0\\
                      & \binoM^t & & \multicolumn{1}{c|}{\vdots} & \vdots & \phantom{-}\vdots \\
                      & & & \multicolumn{1}{c|}{*} & 0 & \phantom{-}0\\\cline{1-4}
                      0 & \dots &  0 & 1 & 1 & -1
                    \end{array}\right)$}.
\end{tabular}
}
\end{minipage}

Both sets of indices $J$ and
$\tilde{J}$ witnessing multistationarity do not contain $P_1$. Then,
from the structure of the matrix
$\sign(\det(\hat{\stoichM}^t_{J\cup \{P_1\}})\det(\hat{\binoM}^t_{J\cup \{P_1\}}))\neq$ 
$\sign(\det(\hat{\stoichM}^t_{\tilde{J}\cup \{P_1\}})\det(\hat{\binoM}^t_{\tilde{J}\cup \{P_1\}}))$, 
which by Theorem~\ref{th:monostationarity} ensures that the
cascade with the drug is multistationary. 
\end{example}

For s-toric MESSI systems we give in Proposition~\ref{prop:rankM} below sufficient conditions 
for the hypothesis in Theorem~\ref{th:monostationarity} that the ranks of $M$ and
$B$ coincide. These conditions are not necessary, but if any of them is not satisfied, 
the ranks might be different.

\begin{proposition}\label{prop:rankM}
Let $G$ be an s-toric MESSI network $G$. Assume that the partition is minimal with $m$ subsets
of core species and the associated digraph $\GE$ has no directed cycles. Then, 
$\rank(\binoMp)=\rank(\stoichMp)=m$.
\end{proposition}

\begin{example}[Necessity of the hypothesis about $\GE$ in Proposition~\ref{prop:rankM}] 
\label{ex:rankM}
If there are directed cycles in $\GE$, we cannot assert that $\rank(\stoichMp)=\rank(\binoM)$. 
Consider, for instance, the following MESSI network without intermediate complexes:
{\small
\[S_0+R_1 \overset{\kappa_1}{\rightarrow} S_1+R_1 \qquad S_1+R_0\overset{\kappa_2}{\rightarrow} S_0+R_0\]
\[P_0+S_0 \overset{\kappa_3}{\rightarrow}  P_1+S_0 \qquad P_1+S_1 \overset{\kappa_4}{\rightarrow} P_0+S_1\]
\[R_0+P_0 \overset{\kappa_5}{\rightarrow} R_1+P_0 \qquad R_1+P_1\overset{\kappa_6}{\rightarrow} R_0+P_1,\]}
where $\Sp$ is the disjoint union of $\Sp^{(1)}= \{S_0, S_1\}$, $\Sp^{(2)}= \{P_0, P_1\}$, and $\Sp^{(3)}= \{R_0, R_1\}$. 
The corresponding digraph $G_2$ equals
\[S_0 \overset{\kappa_1 r_1}{\underset{\kappa_2 r_0}{\rightleftarrows}} S_1 \qquad
P_0 \overset{\kappa_3 s_0}{\underset{\kappa_4 s_1}{\rightleftarrows}} P_1 \qquad
R_0 \overset{\kappa_5 p_0}{\underset{\kappa_6 p_1}{\rightleftarrows}} R_1\]
and the digraph $G_E$ is a cycle:
 {\small 
\begin{tabular}{ll}
   \multirow{2}{*}{\begin{tikzpicture}[ampersand replacement=\&] 
     \matrix (m) [matrix of math nodes, row sep=2.5em, column sep=0.1em, text height=1.5ex, text depth=0.25ex]
      {\Sp^{(1)} \& {\rightarrow} \& \Sp^{(2)} \& {\rightarrow}  \& \Sp^{(3)}. \\};
     \draw[->](m-1-5) to[in=-45,out=-135] 
node[below] (x) {} (m-1-1);
     \end{tikzpicture}} \\
 \end{tabular}
}

\vskip .8cm

We call $s_0, s_1$ the concentrations of $S_0, S_1$ (respectively), $p_0, p_1$  the 
concentrations of $P_0, P_1$, $r_0, r_1$ the concentrations of $R_0, R_1$.
There are three linearly independent conservation relations:
\[ s_0 + s_1 = C_1, \quad p_0 + p_1 = C_2, \quad r_0 + r_1 = C_3.\]
We expect the rank of $B$ to be $3$.
But the system equals
\small{
\[ d s_0/dt = -\kappa_1 s_0 r_1 + \kappa_2 s_1 r_0, \, \,
d p_0/dt = -\kappa_3 s_0 p_0+ \kappa_4 s_1 p_1, \, \, 
 d r_0/dt = -\kappa_5 p_0 r_0 + \kappa_6 p_1 r_1,\]}
and so we can choose $B$ to be the matrix:
$\left(\begin{array}{@{\hskip 1mm}r@{\hskip 2mm}r@{\hskip 2mm}r@{\hskip 2mm}r@{\hskip 2mm}r@{\hskip 2mm}r@{\hskip 1mm}}
           -1 & 1 &  0 & 0 &  1 & -1\\
           -1 & 1 & -1 & 1 &  0 & 0\\
            0 & 0 & -1 & 1 & -1 & 1
          \end{array}\right)$,
which has rank $2$.

Assume there exists a positive steady state. Then, we deduce that
\begin{equation}\label{eq:kappa}
 \kappa_1  \kappa_4  \kappa_5 \, = \, \kappa_2 \kappa_3 \kappa_6.
\end{equation}
So, when~\eqref{eq:kappa} is not satisfied, there are no positive steady states and when it is satisfied, 
any of the three steady state equations is a consequence of the other two, and when we intersect 
with the linear variety defined by the conservation relations, we get a variety of dimension $1$, 
with an infinite number of positive steady states (there are $5$ equations in $6$ variables).
\end{example}

If a consistent toric MESSI system is not monostationary,
we can effectively construct two
different steady states $\bx^1$ and $\bx^2$ and a reaction rate constant vector $\Ka$
that witness multistationarity based on item (3) in the statement of 
Theorem~\ref{th:monostationarity}, following the arguments in~\cite{PM12}
(see also \cite{CoFlRa08,Fein95DefOne}).

\begin{theorem}\label{th:multistationarity}
Let $G$ be a consistent MESSI network which satisfies the hypotheses of 
Theorem~\ref{th:monostationarity}, such that the associated system is toric and it is not monostationary. 
Then, for any choice
of  $\ww \in \stoich, \vv \in \bino^\bot$ in the same orthant,
the positive vectors $\bx^1$ and $\bx^2$ defined as
 \begin{align*}
  \left(x^1_i\right)_{i=1,\, \ldots,\, s} ~&=~
  \left\lbrace\begin{array}{lll}
               \frac{w_{i}}{e^{v_{i}}-1}, & & \text{if }v_{i} \neq 0 \\
	       \text{ any } \bar{x}_i>0, & & \text{otherwise,}
              \end{array}\right.\\
  \mathbf{x}^2 ~&=~ \diag(e^{\vv})\, \mathbf{x}^1
 \end{align*}
 are two different steady states of the given toric MESSI system for any vector of rate
 constants $\Ka$ which is a positive solution of the linear system
 $f(\bx^1,\Ka)=0$, with $f(\bx^1,\Ka)$ as in~\eqref{CRN}. 
\end{theorem}

\subsection*{An algorithm to find different steady states in multistationary toric MESSI systems}

\medskip

We present here an algorithm based on Theorems~\ref{th:monostationarity} and~\ref{th:multistationarity} 
which checks whether a consistent toric MESSI system has the capacity 
for multistationarity. In this case, it looks for orthants where $\stoich$ and $\bino^\bot$ meet 
and finds two different steady states in the same stoichiometric compatibility class,
 together with a corresponding set of 
reaction constants (based on \cite{CoFlRa08,Fein95DefOne,PM12}). 

The algorithm to find these orthants relies on the theory of oriented matroids \cite{Bjetal99,RiZi97,Rock69}.
Recall that the \emph{support} of a vector is defined as the set of its
nonzero coordinates. A circuit of a real matrix $A$ is a nonzero element 
$r \in \mathrm{rowspan}(A)$ with minimal support (with respect to inclusion). 
Given an orthant $\mathcal{O}$ (resp. a vector $v$), a circuit $r$ is said to be 
\emph{conformal} to $\mathcal O$ (resp., $v$) if for any index $i$ {\it in its support}, 
 $\mathrm{sign}(r_i)=\mathcal{O}_i$ (resp.,  $\mathrm{sign}(r_i)=\mathrm{sign}(v_i)$).
A key result is that every vector $v \in \mathrm{rowspan}(A)$ is a nonnegative sum 
of circuits conformal to $v$ \cite{Rock69}. All the circuits of $A$ can be described in terms of 
vectors of maximal minors of $A$ (see Lemma~\ref{lem:circuits} in the Appendix) and one can thus compute all 
orthants containing vectors in $\mathrm{rowspan}(A)$ as those orthants $\mathcal O$ 
whose support equals the union of the supports of the circuits conformal to $\mathcal O$.  
These arguments also allow us to check the consistency of a given network, that is, whether 
there is a positive element in the kernel of a matrix with columns given by the reaction vectors $y'-y$.

\begin{algorithm}
 \label{algo}
Given a consistent toric MESSI system with network $G$, the following procedure finds, if they exist,
multistationarity parameters $\Ka$ or decides that the system is monostationary.
\begin{itemize}
 \item[Input:] A  toric MESSI network $G$. \\ \vspace{-0.3cm}
 \item[Step 0:] Compute matrices $\stoichMp$ (or $M$) and $\binoM$ (or $B^\bot$) for $G$. 
 \item[Step 1:] Compute $\Sigma^\bot$ (or any of the sets $\Sigma,\Sigma_\bot, \Sigma^\bot_\bot$). 
 Check if $\Sigma^\bot$ is mixed.
 If it is unmixed, stop and assert that the system is monostationary.
 \item[Step 2:] Compute the circuits for $\binoMp$ 
 and find an orthant whose support equals the union of the circuits conformal to it.
 \item[Step 3:] For the orthant computed in Step 2,
 check if there is a conformal circuit of $\stoichM$ contained in this orthant. In this case,
 check whether its support equals the union of the circuits of $\stoichM$ conformal to it. 
 Otherwise, ignore it, and go back to Step 2.
 \item[Step 4:] For each orthant $\mathcal O$ with $\stoich\cap \mathcal{O}\neq \emptyset$
 and $\bino^\bot\cap\mathcal{O}\neq \emptyset$,  keep the conformal circuits.
 \item[Step 5:] Build vectors $\vv\in \bino^\bot$ and $\ww \in \stoich$, 
 for example, as the sum of the corresponding conformal circuits.
 \item[Step 6:] Output $\bx^1, \bx^2$ and $\Ka$ that witness multistationarity, as in 
 Theorem~\ref{th:multistationarity}. 
\end{itemize}
\end{algorithm}

Efficiency can certainly be improved at any step of the algorithm, mainly to avoid unnecessary 
computations. The rows of $\stoichMp$ usually present some nice structure that minimizes the 
search for orthants containing a circuit, because in the conditions of 
Theorem~\ref{th:conservations} all columns corresponding to the same set in the partition of 
the species are equal, which produces many zero minors that can be predicted. In Step 5, 
infinitely many different choices of $\vv$ and $\ww$ can be obtained by considering positive 
linear combinations of all circuits which are conformal to the orthant $\mathcal{O}$ (one 
circuit per support).

We implemented this algorithm in Octave~\cite{octave} for the cascades in Figure~\ref{fig:sadi}. 
In the multistationary case of only one phosphatase $F$, we obtained two different orthants 
$\mathcal O_1, \mathcal O_2$ where $\stoich$ and  $\bino^\bot$ meet. 
In both cases, we computed for $i=1,2$  a choice of corresponding rate constants $\Ka(i)$ and 
two steady states $\bx^1(i)$ and $\bx^2(i)$ in the same stoichiometric compatibility class. 
We ordered the species $S_0$, $S_1$, $P_0$, $P_1$, $ES0$, $FS1$, $S1P0$, $FP1$, $E$, $F$.
 We considered in both cases two sets of initial conditions (on the same 
stoichiometric compatibility class); first we set initial states 
$S_0=S_{tot}$, $P_0=P_{tot}$, $E=E_{tot}$, $F=F_{tot}$ and  then  initial states $S_0=S_{tot}$, 
$P_1=P_{tot}$, $E=E_{tot}$, $F=F_{tot}$, and all the other species equal to zero. 
We simulated the system and we depicted the output
in Figure~\ref{fig:plot}, which confirms the occurrence of two stoichiometrically compatible 
steady states for $\Ka(1)$ and $\Ka(2)$.
Approximate values are as follows: 
\[\Ka(1)\cong(25.46,0.86,0.86,11,0.86,0.86,0.14,0.21,0.21,37.47,0.21,0.21), \]
\[\bx^1(1)\cong(0.037,3.47,4.07,1.02,1.16,1.16,4.75,4.75,2.1,0.052),\]
\[\bx^2(1)\cong(2.04,0.47,11.07,0.019,3.16,3.16,1.7,1.75,0.1,1.05),\]
and
\[\Ka(2)\cong(101.86,1.72,1.72,33,0.86,0.86,37.47,0.13,0.13,0.42,0.63,0.63),\] 
\[\bx^1(2)\cong(0.019,0.052,1.02,4.07,0.58,1.16,7.91,1.58,1.05,1.16),\]
\[\bx^2(2)\cong(1.02,1.05,0.019,11.07,1.58,3.16,2.9,0.58,0.052,0.16).\]

\begin{figure}[t!]
  \begin{tabular}{ll}
\includegraphics[scale=0.38,trim={2.4cm 6.4cm 1cm 6cm},clip]{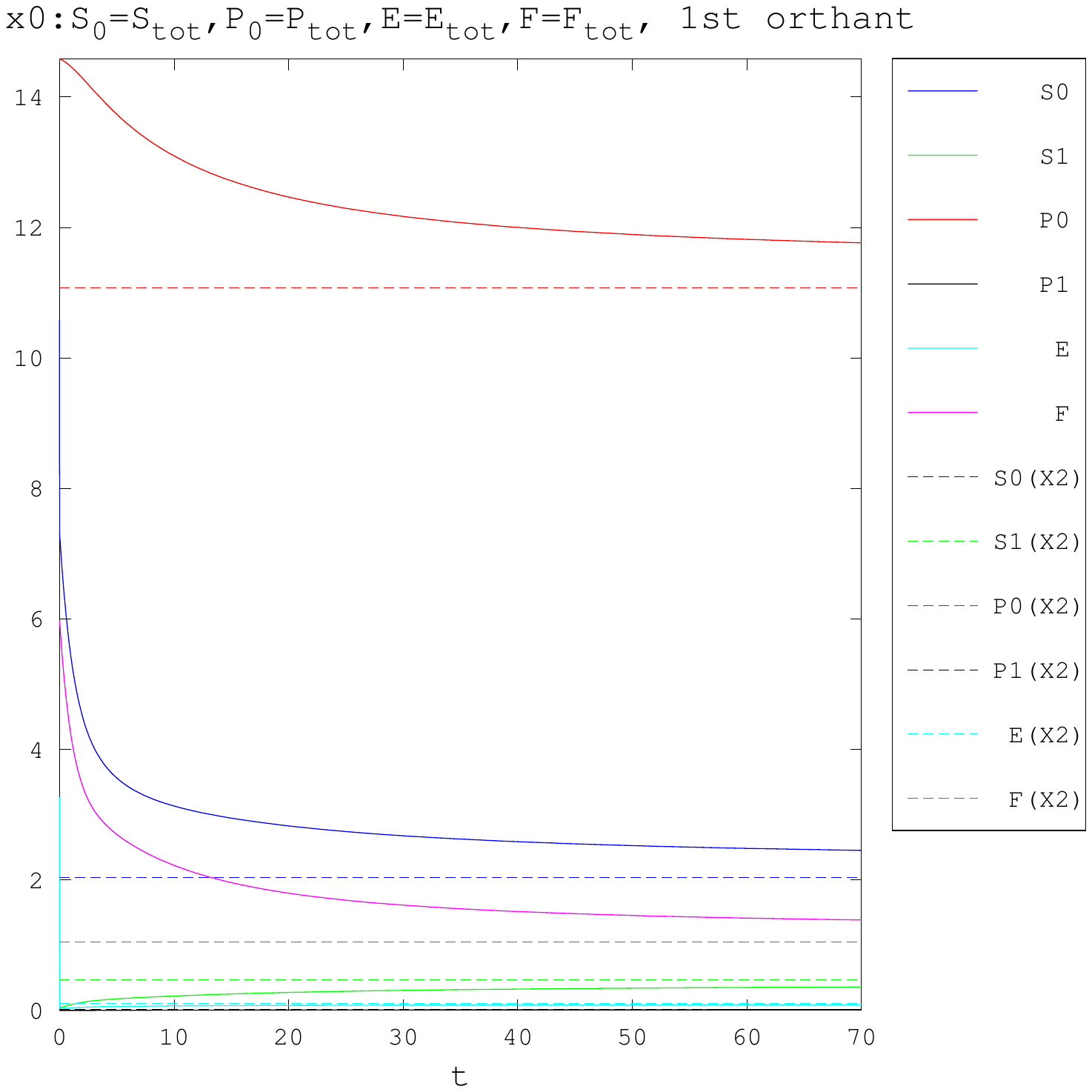} & 
\includegraphics[scale=0.38,trim={2.4cm 6.4cm 1cm 6cm},clip]{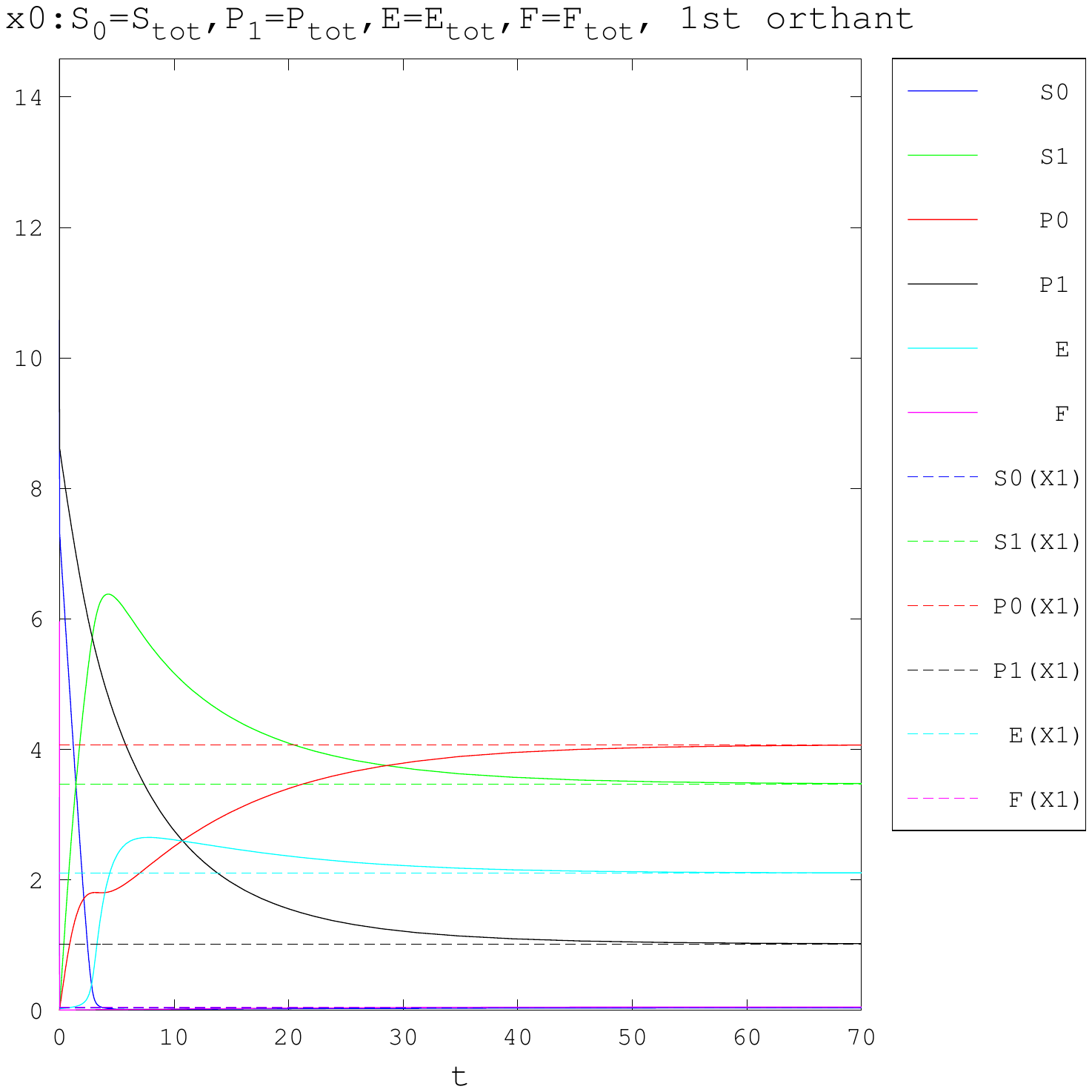} \\
\includegraphics[scale=0.38,trim={2.4cm 6.4cm 1cm 6cm},clip]{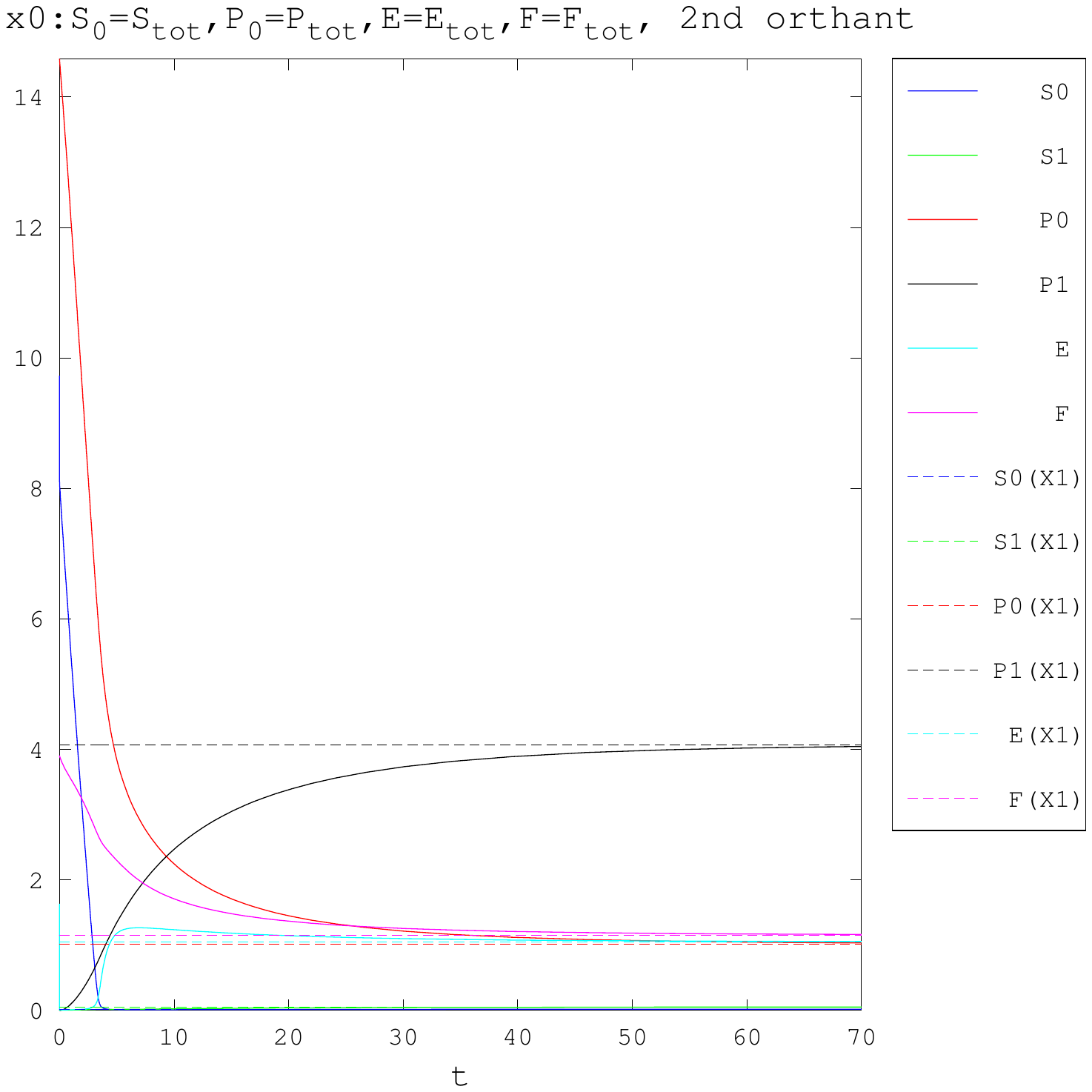} & 
\includegraphics[scale=0.38,trim={2.4cm 6.4cm 1cm 6cm},clip]{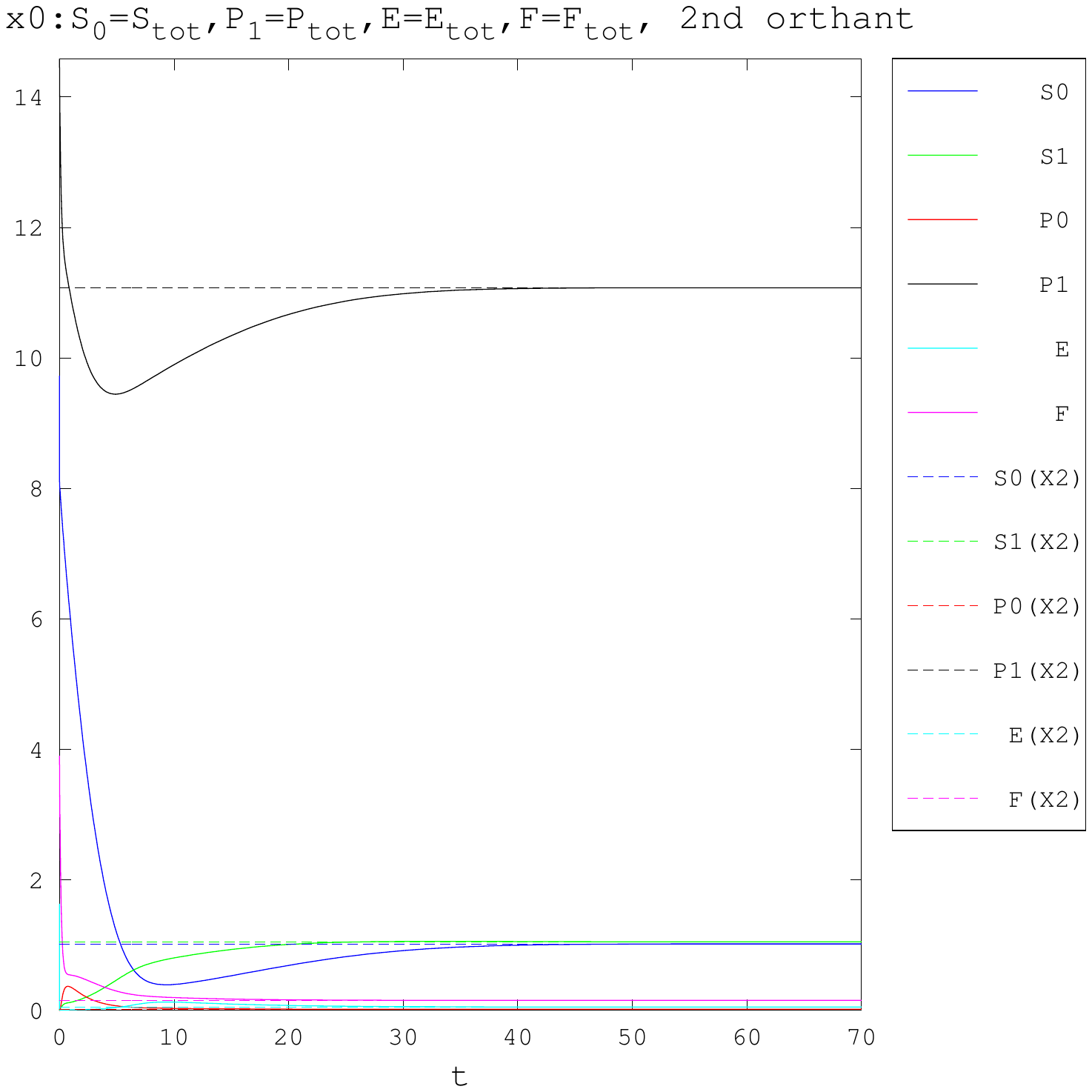}
\end{tabular}
\caption{Witnesses for multistationarity of the phosphorylation cascade with one phosphatase $F$,
  with reaction constants and total amounts obtained from  two orthants $\mathcal{O}_1, \mathcal{O}_2$
 given by  Algorithm~\ref{algo}.
Upper plots depict the two different steady states constructed from $\mathcal{O}_1$ (dashed lines) along with the simulated 
trajectories of $S_0$, $S_1$, $P_0$, $P_1$, $E$ and $F$. The initial state on the left
 is {\small $S_0=S_{tot}$, $P_0=P_{tot}$, $E=E_{tot}$, $F=F_{tot}$}, 
and the initial state on the right is {\small $S_0=S_{tot}$, $P_1=P_{tot}$, $E=E_{tot}$, $F=F_{tot}$}. 
The lower plots correspond to $\mathcal{O}_2$, with the same initial conditions.
We used the function \texttt{ode23s}, from the package \emph{odepkg} version 0.8.5 in Octave~\cite{octave}.}\label{fig:plot}
\end{figure}

\section{Discussion}

\medskip

Our contribution to the study of many different important biological systems modeled with mass-action kinetics
is the identification of a common underlying structure in quite diverse networks. We call this
a MESSI structure, since it describes Modifications of type
Enzyme-Substrate or Swap with Intermediates. The mathematical formulation of the distinguished properties
of MESSI biological systems allows us
to prove general results on their dynamics from the structure of the network.
We give very precise hypotheses that ensure the validity of our statements and which can be easily
verified in common networks of biological interest.

It is important to observe that all the conditions and hypotheses in our paper can be algorithmically
checked. In particular, it is possible to devise an algorithm to check whether a given network has a
MESSI structure,  to prove that a given partition is minimal,
to construct the associated digraphs and networks, including the corresponding labels, and to check the hypotheses of all
our statements. The construction of the rational parametrization in Theorem~\ref{th:par} is also algorithmic.
 Note also that the sufficient conditions which ensure persistence in Theorem~\ref{th:bss} are independent
of the conditions to have a toric MESSI system or even an s-toric MESSI system, including the criterion for
multistationarity given in Theorem~\ref{th:monostationarity}. However, the hypotheses in Proposition~\ref{prop:rankM}
to ensure the validity of the hypotheses in Theorem~\ref{th:monostationarity} also imply persistence. This does not mean
that multistationarity is related to persistence,  but when there are boundary steady states the hypotheses of Theorem~\ref{th:monostationarity}
should be verified in an ad hoc manner.

%

\vspace{2cm}

\appendix

\section{Proofs}\label{sec:pf}

We assume the reader is familiar with the notion of the Laplacian $\Lap(G)$ of a 
digraph $G$ and its main properties. One key observation is that mass-action 
kinetics associated with a linear digraph $G$ with variables $x=(x_1, \dots, x_s)$ 
equals $\dot x =\Lap(G) x$. A second key observation is that the fact that the rows 
of $\Lap(G)$ add up to zero translates into $\sum_{i=1}^s \dot x_i=0$, and so 
$\sum_{i=1}^s x_i$ is a conserved quantity. 
The last key observation is that when $G$ is strongly connected, the kernel of $\Lap(G)$ 
has dimension one and there is a known generator $\rho(G)$ with positive entries described 
as follows.
Recall that an $i$-tree $T$ of a graph is a spanning tree where the $i$th vertex is its unique 
sink (equivalently, the $i$th is the only vertex of the tree with no edges leaving from it), 
and  we call $c^{T}$ the product of the labels of all the edges of $T$. Then, the $i$th 
coordinate of $\rho(G)$ equals
\begin{equation}\label{eq:rhoi}
\rho(G)_i=\underset{T\; an \; i-tree}{\sum}c^{T}.
\end{equation}
We refer the reader to~\cite{MiGu13,Tutte} for a detailed account.

\subsection*{Conservation Relations and Persistence}\label{ap:consrel}
 
In order to prove Theorem~\ref{th:conservations}, we need to first introduce a remark.
We call $\stoich_1$ the stoichiometric subspace of the biochemical network defined by 
the associated digraph $G_1$ of a MESSI reaction network $G$ (with stoichiometric subspace 
$\stoich$). We denote by  
$\tilde{\stoich_1}=\{(0,\dots,0,w)\in \R^{p+n}: w\in\stoich_1\subseteq \R^n\}$
the lifting of $\stoich_1$ to $\R^{p+n}$.

\begin{remark}\label{rem:n+p}
 With the previous notation, the following equality of dimensions is an immediate 
 consequence of Lemma~1 in the ESM of~\cite{fw13}:
\begin{equation}\label{eq:dim}
 \dim(\stoich)=\dim(\stoich_1)+p.
\end{equation}
 \end{remark}

We will also need Lemma~\ref{lem:min} in the main text. 

\begin{proof}[Proof of Lemma~\ref{lem:min}]
Each vertex in the associated digraph $G_2$ to the digraph $G$ is labeled by 
only one species. If one species of $\Sp^{(\alpha)}$ appears on a vertex of $G_2$,
by $({\mathcal R}_2)$ and $({\mathcal R}_3)$ and the construction of $G_2$, all the 
species in the vertices of the corresponding connected component of $G_2$ belong to 
the same $\Sp^{(\alpha)}$.
Moreover, if two core species $X_i, X_h$ in the same subset $\Sp^{(\alpha)}$ correspond 
to different connected components of $G_2$, then for any complex $y_{ij}$ containing $X_i$ 
and any complex $y_{h \ell}$ containing $X_h$, the relation $y_{ij} \uri y_{h \ell}$ does 
{\it not}  hold. It follows that we can refine each subset $\Sp^{(\alpha)}$ as the disjoint 
union of the subsets of species in each connected component of $G_2$ which consists of species 
in $\Sp^{(\alpha)}$, and no further refinement is possible if the set of intermediate species 
is maximal.
\end{proof}

We are ready to prove Theorem~\ref{th:conservations}. We will mainly adapt the results in~\cite{fw13} (Theorem~2.1) to our 
setting.

\begin{proof}[Proof of Theorem~\ref{th:conservations}]
Given a chemical reaction network $G$ and a partition of the set of species $\Sp$ that leads to a
MESSI system with the given complexes and reactions, consider the mass-action system defined by $G_1$, 
with species $X_1,\dots, X_n$.
By Theorem~2.1 in \cite{fw13},  the conservation relations in $G$ are in one-to-one correspondence with 
the conservation relations of $G_1$ in an explicit way that we detail below after our hypotheses. 
Recall that by Remark~\ref{rem:G2}, the associated graph $G_2$ determines the same 
equations.

Fix $\alpha \ge 1$. As we remarked in the proof of Lemma~\ref{lem:min},
each subset $\Sp^{(\alpha)}$ coincides with the variables in the vertices of some of the
connected components of the associated digraph $G_2$. Given such a connected component $H$, 
let $\Sp^{(\alpha)}_{H}$ be its set of vertex labels. As  $G_2$ is a linear digraph, $H$ is also linear and 
so the matrix of the associated (linear) system is given by  its Laplacian $\Lap(H)$. Therefore, 
the sum of its rows equals zero, which means that $\sum_{X_i\in S^{(\alpha)}_H}\dot x_i=0$ 
and a fortiori $\sum_{X_i\in S^{(\alpha)}}\dot x_i=0$,
for the mass-action system defined by $G_1$. We find now the corresponding linear combination which 
includes the concentrations of the intermediate species by adapting Lemma~1 in the ESM of~\cite{fw13}.

Let $\omega^\alpha \in \{0,1\}^n$ be the characteristic vector of $\Sp^{(\alpha)}$, so that
$\langle \omega^\alpha, \dot x \rangle = \sum_{X_i\in S^{(\alpha)}}\dot x_i$.
For any complex $y^j$ of $G_1$,  we know from $({\mathcal R}_2)$ and $({\mathcal R}_3)$ 
that it has at most one species in $\Sp^{(\alpha)}$. Then, 
\[\omega^\alpha\cdot y^j=\left\lbrace\begin{array}{ll}
                        1 & \text{if there is a species of } \Sp^{(\alpha)} \text{ in }y^j\\
                        0 & \text{otherwise.}
                       \end{array}\right.\]
Define the $(p+n)$-vector:
\[\widetilde{\omega}^\alpha_i=\left\lbrace\begin{array}{ll}
            \omega^\alpha_i & \text{for } i=p+1,\dots,p+n\\
                        1 & \text{if } i\in \intal(\alpha)\\
                        0 & \text{otherwise,}
                       \end{array}\right.\]
where $\intal(\alpha)$ is as in~\eqref{eq:intal}. Lemma~1 in the ESM of~\cite{fw13} asserts precisely
that the linear form defined by $\widetilde{\omega}^\alpha$ 
leads to the conservation of the whole network associated with the linear form defined
by $\omega^\alpha$ on the variables in $\Sp_1=\Sp \setminus \Sp^{(0)}$. But this linear form 
is precisely $\ell_\alpha$, as we wanted to prove.
Since we are assuming that all species participate in at least one reaction and intermediate 
species satisfy condition~$(\cond)$, we have that $\Sp^{(0)} = \cup_{\alpha=1}^m \Sp\intal(\alpha)$. 
Therefore, all coefficients of the conservation relation $\sum_{\alpha=1}^m 
\ell_\alpha$ are positive and we get that any MESSI system is conservative.

To see the second part of the statement, note that $\ell_{1}, \dots, \ell_{m}$ define linearly independent
conservation relations and so $\dim(\stoich^\bot)=s-\dim(\stoich)\ge m$. It only remains to prove that,
if $G$ has no swaps, then $s-\dim(\stoich)\leq m$. By
Remark~\ref{rem:n+p} it holds that $\dim(\stoich)=\dim(\tilde{\stoich}_1)+p$  because clearly  
$\dim(\stoich_1)=\dim(\tilde{\stoich}_1)$. It is then enough to show that $\dim(\tilde{\stoich}_1)\ge n-m$.
If $X_i+X_j\to X_\ell + X_k$ in $G_1$, and there are no swaps in $G$, either $i\in\{\ell,k\}$ or
$j\in\{\ell,k\}$. Assume, without loss of generality, that $j=k$. Then $e_\ell-e_i\in \stoich_1$,
for $e_i$ is the $i$th canonical vector of $\R^n$.
As $\Sp$ is minimal, if $X_i, X_\ell \in \Sp^{(\alpha)}$, necessarily $X_i$ and $X_\ell$ belong to the same 
connected component of $G_2$.
Then there is an undirected path between $X_i$ and $X_\ell$ in $G_2$.
By a telescopic sum, as in the proof of Lemma~\ref{lem:S1} below, we have that
each vector $e_\ell-e_i\in \stoich_1$ for each $X_i,X_\ell\in \Sp^{(\alpha)}$.
Fix $X_i\in \Sp^{(\alpha)}$; then for all $\ell \neq i$, $e_\ell-e_i\in \stoich_1$.
This gives us $n_\alpha-1$ linearly independent vectors for each $\alpha\ge 1$, which
are in turn linearly independent from the corresponding vectors obtained from each $\beta$, 
$\beta \neq \alpha$, $1\le \beta\le m$ (when $n_\alpha >1$). Adding over $\alpha \ge 1$, we obtain $n-m$ 
linearly independent vectors in $\stoich$. (Notice that if $\Sp^{(\alpha)}$ is a singleton, $n_\alpha-1=0$.)
Therefore, $\dim(\stoich)\ge p+n-m=s-m$, which is what we wanted to prove. The total number of conservation 
relations in a system is equal to the codimension of the kinetic subspace. If, morover, the kinetic subspace 
equals $\stoich$, then $\dim(\stoich^\bot)=m$, as claimed.
\end{proof}


We now focus on the occurrence of boundary steady states. Both proofs of Theorem~\ref{th:bss} and 
Proposition~\ref{prop:int} below are  based on the proof of Theorem~3.1 in \cite{fw13} (Theorem~2 in their ESM).

\begin{proof}[Proof of Theorem~\ref{th:bss}]
Assume there is a boundary steady state in some stoichiometric compatibility 
class that intersects the positive orthant. 

Following the proof of Theorem~2 in the ESM of \cite{fw13}, it can be seen that 
at steady state the concentration of an intermediate species $u_k$ is a nonnegative 
linear combination of monomials in the concentrations of the core species in the complexes 
that react via intermediates to it. Then, if there is an intermediate species $U_k$ such 
that $u_k=0$  at steady state, there is at least one core species (in a core complex that 
reacts via intermediates to $U_k$) that vanishes at steady state.
Therefore, if there is a boundary steady state, there is a core species $X_i$ 
such that $x_i=0$ at steady state.

By Lemma~\ref{lem:min},  we can refine the given MESSI structure in such a way that subsets 
of core species are in bijection with the connected components of $G_2$. In order to avoid 
unnecessary notation, we will assume in what follows that the partition is minimal.
Recall that a vertex in a directed graph has \emph{indegree zero} if it is not the
head of any directed edge. Let us define the subsets of indices
\begin{align*}
  L_0=&\{\beta \ge 1 : \text{indegree of }\Sp^{(\beta)}\text{ is }0\},\\
  L_k=&\{\beta \ge 1: \text{for any edge }  \Sp^{(\gamma)}\to\Sp^{(\beta)}\text{ in }\GE \text{ it holds that }
  \gamma \in L_t, \text{ with }  t<k \}\backslash \underset{t=0}{\overset{k-1}{\bigcup}} L_t,  k \ge1.
\end{align*}
The main observation that makes the following inductive argument work is that as  $\Sp$ 
is finite and there are no directed cycles in $\GE$, there must
exist a subset $\Sp^{(\beta)}$ with $1\le \beta\le m$ such that its indegree in $\GE$ is zero. 
This means that $L_0\neq \emptyset$.

Let $\ell \ge 0$ be minimal with the property that there exist $\alpha \in L_\ell$ and a core 
species $X_i \in \Sp^{(\alpha)}$ such that $x_i=0$ at steady state. Denote by $H_\alpha$ the 
connected component of $G_2$ with vertices the species in $\Sp^{(\alpha)}$.
Let $\rho(H_\alpha)$ be the generator of the kernel of $\Lap(H_\alpha)$ as in~\eqref{eq:rhoi}. 
Its entries are nonnegative sums of terms involving the rate constants $\tau$ and concentrations 
of species in $L_j$ with $j < \ell$. Then, $\rho(H_\alpha)$ has nonzero coordinates since 
$H_\alpha$ is strongly connected because $G_2$ is weakly reversible and $\ell$ is minimal.
Moreover, the following equation is satisfied at steady state for any $X_j \in \Sp^{(\alpha)}$:
\begin{equation}\label{eq:binrho}
 \rho(H_\alpha)_j \,  x_i- \rho(H_\alpha)_i \, x_{j}=0.
\end{equation}
Then the corresponding concentrations $x_j$ vanish at
steady state for any $X_j \in \Sp^{(\alpha)}$.
Take $k \in \intal(\alpha)$. The concentration of the intermediate species $u_k$ is a nonnegative 
linear combination of monomials in the concentrations of the core species that react via intermediates 
to it. By condition $(\cond)$ and rule $({\mathcal R}_3)$, any such monomial contains one variable 
indexed by a species in $\Sp^{(\alpha)}$.
As  $x_j=0$ for all $j \in \Sp^{(\alpha)}$ we get that $u_k=0$. 
This gives a contradiction by \eqref{eq:consalpha} in Theorem~\ref{th:conservations} since 
$C_\alpha$ is a nonzero constant.

As MESSI systems are conservative, the existence of nonnegative steady states is guaranteed by 
fixed-point arguments.  Indeed, a version of the Brouwer fixed-point theorem ensures that a nonnegative 
steady state exists in each compatibility class.  As the system has no boundary steady states,
we deduce the existence of a positive steady state in each compatibility class, and, in 
particular, the consistency of the system.
\end{proof}

\subsection*{Parametrizing the steady states}\label{sec:parproofs}

We  first prove  the existence of rational parametrizations under the hypotheses
of Theorem~\ref{th:par}.

\begin{proof}[Proof of Theorem~\ref{th:par}]
The arguments of the proof are similar to those in the proof 
of Theorem~\ref{th:bss}. Again, we will assume that the partition is
minimal to ease the notation.
Recall the sets $L_k$ in that proof and the crucial
remark that $L_0\neq \emptyset$ because the graph $\GE$ has no directed cycles.

For each $\alpha\geq 1$, fix $X_{i_\alpha}\in\Sp^{(\alpha)}$. Because of the minimality of the partition,  
any other $X_i \in \Sp^{(\alpha)}$ lies in the connected component $H_\alpha$ of $G_2$ containing 
$X_{i_\alpha}$. We can then parametrize all the species in $\Sp^{(\alpha)}$ for $\alpha\in L_k$ in 
terms of $x_{i_\alpha}$ and the species in $L_j$ for $j<k$, recursively using~\eqref{eq:binrho} to write
  \[x_i=\dfrac{\rho(H_\alpha)_i}{\rho(H_\alpha)_{i_\alpha}}x_{i_\alpha}\]
at steady state.
Moreover,  the concentrations of intermediate species can be rationally written in terms of all 
$x_{i_\alpha}, \alpha=1,\dots, m$ (see Definition~\ref{def:tau} and Remark~\ref{rem:tauq}). Thus, 
$\dim(V_f \cap \R_{>0}^s)=m$.
The last equality $\dim \stoich^\bot =m$ in the statement follows
from Theorem~\ref{th:conservations} using Remark~\ref{rem:m}.
\end{proof}

We show now that the {\em{positive}} steady states of s-toric MESSI systems can be described 
by binomials, and we postpone the proof of the choice of very explicit binomials when
any pair of nodes in the same component are connected by a single simple path.

\begin{proof}[Proof of Proposition~\ref{prop:int}]
 Following the arguments in \cite{fw13}, we first build a new labeled directed graph
$\widehat{G}$ with node set $\Sp^{(0)}\cup\{*\}$, which consists of collapsing all 
core complexes into the vertex $*$, and labeled directed edges that are obtained 
from hiding the core complexes in the labels. For example, $X_i+X_j\overset{\kappa}{\rightarrow} U_k$
becomes $*\overset{\kappa x_ix_j}{\longrightarrow}U_k$ and $U_k\overset{\kappa'}{\rightarrow} X_i+X_j$
becomes $U_k\overset{\kappa'}{\longrightarrow}*$. This new graph is linear and satisfies that
$\dot{\bu}=0$ is equivalent to $\Lap(\widehat{G})\, \widetilde{\bu}=0$, where
$\widetilde{\bu}=(u_1, \dots, u_p,1)^t$ (this last coordinate stands for 
``the concentration'' of the node $*$).
It is important to notice that the graph $\widehat{G}$ is strongly connected by condition~$(\cond)$.

Then, at steady state we obtain that $\widetilde{\bu}$ is proportional to the vector
$\rho_{\widetilde{G}}= (\rho_1, \dots, \rho_p, \rho)$ defined in~\eqref{eq:rhoi}, 
so that $u_k=\rho_k/\rho$ for any $k=1, \dots, p$.  It is straightforward to check 
that every $*$-tree involves labels in $\Q[\kappa]$. 
On the other hand, for every $U_k$, as by condition $(\cond')$ there is a unique
core complex $y_{i_kj_k}$ such that $y_{i_kj_k}\uri y_k$, every $k$-tree involves labels in
$\Q[\kappa, x_{i_k}x_{j_k}]$. Moreover, as there must be a path from $*$ to $U_k$ in each $k$-tree,
$x_{i_k}x_{j_k}$ necessarily appears as a label on those trees. Then,
\begin{equation}\label{eq:uk}
u_k= \mu_k \,  x_{i_k}x_{j_k}, \, k=1, \dots p,
\end{equation}
where 
\[\mu_k =\frac{\rho_k}{x_{i_k}x_{j_k}} \frac{1} {\rho} \in \Q(\kappa).\]
\end{proof}

\begin{proof}[Proof of the first part of Theorem~\ref{th:toric_toric}]
 Let  $\bx$ be a positive steady state and $X_i\neq X_j$ in $\Sp^{(\alpha)}$ in 
 the same connected component $H$ of $G_2$. Let $\rho(H)$ be the explicit generator 
 of the kernel of $\Lap(H)$ as in~\eqref{eq:rhoi}. Then, as in~\eqref{eq:binrho}, 
 $\rho(H)_j x_i- \rho(H)_i x_j=0$.
 Fix a $j$-tree $T_0$. The product of the labels $c^{T_0}$ of all the edges
 in $T_0$ is equal to a monomial $x^{\gamma_j}$ times a polynomial in
 the rate constants $\tau$. For any other $j$-tree $T$,  condition ($\cond'''$)
 ensures that $c^T = \mu_T(\tau) \, c^{T_0}$, with $\mu_T \in \Q(\tau)$. 
 It follows that the quotient of the sum $\rho(H)_j$ by $x^{\gamma_j}$
 lies in $\Q(\tau)$ (and also there exists a monomial $x^{\gamma_i}$ such that 
 $\rho(H)_i/x^{\gamma_i} \in \Q(\tau)$). Call
 \begin{equation}\label{eq:eta}
 \eta_{ij}=\rho(H)_i x^{\gamma_j}/\rho(H)_j x^{\gamma_i}\in \Q(\tau) \subset \Q(\Ka).
 \end{equation}
 Then, 
$x^{\gamma_j}x_i-\eta_{ij}x^{\gamma_i}x_j=0$.
 Combining this with~\eqref{eq:uk}, the {\it positive} steady states can be described 
 by the binomials:
 \begin{align}\label{eq:b}
  & u_k-\mu_k \bx^{\varphi(k)} \text{ for each intermediate species }U_k\\ \label{eq:b1}
  & x^{\gamma_j}x_i-\eta_{ij}x^{\gamma_i}x_j \text{ if }  X_i , X_j \text{ lie  in the same connected component of } G_2.
 \end{align}
 We can fix one species $X_{i_h}$ in each connected component $H$
 of $G_2$ and consider the binomial equations of the form in~\eqref{eq:b1} where $i=i_h$.   
 There are $p$ further binomial equations in~\eqref{eq:b}. These $p + n-m'= s-m'$ binomial 
 equations cut out the positive steady states.
\end{proof}

To prove the second part of Theorem~\ref{th:toric_toric},
we first need a combinatorial lemma.
 
\begin{lemma}\label{lem:unique_tree}
Assume $H$ is a digraph with the property that there is a unique simple path $P_{ij}$ 
from any node $X_i$ to any node $X_j$  in the same connected component of $H$.
Then the following hold:
\begin{enumerate}
\item[(i)] For each vertex $X_i$ of $H$ there is only one $i$-tree, denoted by $T_i$.
\item[(ii)] Let  $X_i\overset{\tau x_h}{\longrightarrow}X_j$ be  an edge in $H$. Then, 
$T_i$ is obtained from $T_j$ by deleting the edge
 $X_i\overset{\tau x_h}{\longrightarrow}X_j$ and adding the edge
 $X_j\overset{\tau' x_m}{\longrightarrow}X_\ell$, where $X_\ell$ is such that
 $X_j\overset{\tau' x_m}{\longrightarrow}X_\ell$ is in $P_{ji}$.
\end{enumerate}
\end{lemma}

\begin{proof}
Proof of (i): Let $X_j$ ($j\neq i$) be in the same connected component of $H$ as $X_i$. 
In any $i$-tree there is an edge leaving from $X_j$; otherwise $X_j$ would be another 
sink different from $X_i$. Moreover, there must be a path from $X_j$ to $X_i$ in any 
such $i$-tree. If the path visits some vertex twice (or more times), there would be a 
cycle in the underlying undirected graph of the tree, which is not possible.
Hence, the path is simple. By hypothesis, there is only one choice for this path,
and so there is only one $i$-tree in $H$.

Proof of (ii): 
 Call $T'$ the new digraph obtained from $T_j$ by deleting the edge
 $X_i\overset{\tau x_h}{\longrightarrow}X_j$ and adding the edge
 $X_j\overset{\tau' x_m}{\longrightarrow}X_\ell$. $T'$ still visits
 every vertex of the corresponding connected component of $H$, and the only vertex
 from which no arrows leave is $X_i$. We claim that there are no cycles
 in $T'$. In fact, the only possible cycle in $T'$ must involve the new edge from $X_j$ to $X_\ell$. Then,
 there is a directed path in $T'$ (and therefore in $H$) from $X_\ell$ to $X_j$. Moreover,
 as the paths in $T_j$ are simple, this path from $X_\ell$ to $X_j$ in $T'$ is simple. But in
 $H$ there is another simple path $P_{\ell i}\cup\{X_i\to X_j\}$ from $X_\ell$ to $X_j$,
 which is different from the one obtained in $T'$ since the edge $X_i\to X_j$ does not exist in $T'$.
This is a contradiction since by assumption there is only one simple path in $H$ from $X_\ell$ to $X_j$. Then, $T'=T_i$.
\end{proof}

\begin{proof}[Proof of the second part of Theorem~\ref{th:toric_toric}]
 If there is a unique simple path $P_{ij}$ from each $X_i$ to each $X_j$ in the same connected component
 of $G_2$, and $X_i\overset{\tau x_h}{\longrightarrow}X_j$ is in $G_2$, the binomial in~\eqref{eq:b}
 involves the edges on $T_i$ and the edges  on $T_j$. But, from Lemma~\ref{lem:unique_tree}, $T_i$ and $T_j$ only differ in the 
edges
 $X_i\overset{\tau x_h}{\longrightarrow}X_j$ and
 $X_j\overset{\tau' x_m}{\longrightarrow}X_\ell$, where $X_\ell$ is such that
 $X_j\overset{\tau' x_m}{\longrightarrow}X_\ell$ is in $P_{ji}$.
Then, after taking out a monomial, the following binomials define the positive steady states:
 
 \begin{minipage}{0.7\textwidth}
\begin{align*}
  & u_k-\mu_k \bx^{\varphi(k)}\text{ for each intermediate species }U_k\\
  & \tau x_h x_i-\tau'x_m x_j \text{ if }  X_i\overset{\tau x_h}{\longrightarrow}X_j \text{ in } G_2^\circ
  \text{ and } X_j\overset{\tau' x_m}{\longrightarrow}X_\ell \text{ is in } P_{ji}.
 \end{align*}
\end{minipage}
\begin{minipage}{0.25\textwidth}
{\footnotesize
\begin{tikzpicture}
  \matrix (m) [matrix of math nodes, row sep=0.5em,ampersand replacement=\&,
    column sep=0.5em]{
    X_i \& \overset{\tau x_h}{\longrightarrow} \& X_j \\
    \& \& \\
     \& \& X_\ell\\};
     \draw[->,dashed](m-3-3) to[in=-55,out=190] node[below] (x) {} (m-1-1);
     \draw[->](m-1-3) to node[right] (x) {\tiny $\tau' x_m$} (m-3-3);
 \end{tikzpicture}
 }
 \end{minipage}
 \vspace{1mm}
\end{proof}

\subsection*{Toric MESSI systems and Multistationarity}\label{sec:toricmss}

We will prove Theorem~\ref{th:monostationarity} 
by adapting Proposition~3.9 and Corollary~2.15 in \cite{mfrcsd13} and Theorem~5.5 
in \cite{PM12} to our setting.
We recall that a chemical reaction system has the capacity for multistationarity if
there exists a choice of rate constants such that there are two or more {\it positive}  
steady states in one stoichiometric compatibility class
$(x^0+\stoich) \cap \mathbb{R}^s_{\geq 0}$ for some initial state
$x^0 \in \mathbb{R}^s_{\geq 0}$ (and it is monostationary otherwise).

\begin{remark}\label{rem:inj}
Consider a toric MESSI system whose positive steady states can be described by binomial equations of
the form $x^{y'}-\eta x^{y}=0$. Equivalently, the positive steady states of the toric MESSI system 
can be described by the monomial equations $x^{y'-y}=\eta$, where we consider Laurent monomials.
We construct now a matrix $\binoM$ whose columns form a basis of the subspace $\bino$ generated
by these difference vectors $y'-y$, and also the monomial map $x\mapsto x^\binoM$,
where $(x^\binoM)_j=x^{\binoM_j}=x_1^{\binoM_{1j}}\cdot \ldots \cdot x_s^{\binoM_{sj}}$,
for each column $\binoM_j$ of $\binoM$. Then  $x^*$ is a positive steady state of the
system if and only if ${x^*}^{\binoM}=\tilde{\eta}$ for an appropriate vector $\tilde{\eta}$.
Thus, the system is monostationary \emph{for any choice of rate constants}
if and only if the monomial map $x\mapsto x^\binoM$ is injective on each stoichiometric compatibility class
$(x^0+\stoich)\cap \R_{>0}^s$ for every $x^0\in \R_{>0}^s$.
\end{remark}

\begin{proof}[Proof of Theorem~\ref{th:monostationarity}]
Under the hypotheses in the statement, we want to prove the equivalence of the assertions:
\begin{enumerate}
 \item[(i)] \label{itm:monoa} The associated MESSI system is monostationary.
 \item[(ii)] \label{itm:unmixeda} The signs sets $\Sigma, \Sigma^\bot,  
\Sigma_\bot,\Sigma^\bot_\bot$ are unmixed.
 \item[(iii)] \label{itm:orthanta} For all orthants $\mathcal{O}\in\{-1,0,1\}^s,
 \mathcal{O}\neq\mathbf{0}$, either $\stoich\cap \mathcal{O}= \emptyset$
 or $\bino^\bot\cap\mathcal{O}= \emptyset$.
\end{enumerate}
We first prove 
(i) $\Leftrightarrow$ (ii) by adapting the results in
\cite{mfrcsd13}. We will see that (i) and (ii) are both equivalent to
\begin{equation*}
 \{\sign(v): v\in \ker(\binoM^t)\}\cap\{\sign(v): v\in \stoich\}=\{0\},
\end{equation*}
where $(\sign(v))_i=\sign(v_i)$ for $i=1,\dots,s$. This is also equivalent 
by the definition of $\bino^\bot$ to
\begin{equation}\label{eq:signos}
 \{\sign(v): v\in \bino^\bot\}\cap\{\sign(v): v\in \stoich\}=\{0\}.
\end{equation}
By Remark~\ref{rem:inj}, (i) is equivalent to the injectivity of the map $x\mapsto x^\binoM$
on each stoichiometric compatibility class $(x^0+\stoich)\cap \R_{>0}^s$. We deduce
from Proposition~3.9 in \cite{mfrcsd13} that  (i) is equivalent to~\eqref{eq:signos}.
Previously, in Corollary~2.15 the authors had proved that~\eqref{eq:signos} is in turn equivalent to
asking that for all $J\subseteq [s]$, $\#J=s-d=\rank(\binoM)=\rank(\stoichM)$,
$\det(\binoM_J)\det(\stoichM_J)$ is either zero or has the same sign as all other nonzero products,
and moreover, at least one such product is nonzero. In other words, \eqref{eq:signos}
is equivalent to the set $\Sigma$
being unmixed. By Lemma~\ref{lem:signs}, 
this is equivalent to (ii).
To finish the proof, we just need to show that 
\eqref{eq:signos} $\Leftrightarrow$ (iii), but this is straightforward.
\end{proof}

We now prove Theorem~\ref{th:multistationarity}, and we postpone the proof of Proposition~\ref{prop:rankM},
which needs an ancillary lemma.

\begin{proof}[Proof of Theorem~\ref{th:multistationarity}]
 
By Theorem~\ref{th:monostationarity}, 
if the system is not monostationary, we know that there exists an orthant $\mathcal{O}\in\{-1,0,1\}^s,
\mathcal{O}\neq\mathbf{0}$, such that $\stoich\cap \mathcal{O}\neq \emptyset$ and
$\bino^\bot\cap\mathcal{O}\neq \emptyset$. Then, there exist
 $\ww \in \stoich, \vv \in \bino^\bot$ such that $\sign(\ww)=\sign(\vv)$.
 Inspired by Theorem~5.5 in \cite{PM12},  for any index $i$ not in the support of $\vv$,
we choose any positive real number $h_i$ and we define
 positive vectors $\bx^1$ and $\bx^2$ as follows:
 \begin{align*}
  \left(x^1_i\right)_{i=1,\, \ldots,\, s} ~&=~
    \left\lbrace\begin{array}{lll}
               \frac{w_{i}}{e^{v_{i}}-1}, & & \text{if }v_{i} \neq 0 \\
	       h_i, & & \text{otherwise,}
              \end{array}\right.\\
  {x}^2 ~&=~ \textrm{diag}(e^{\vv})\, {x}^1
 \end{align*}
where ``$e^x$'' for a vector $x \in \R_{>0}^s$ denotes the vector
$(e^{x_1}, e^{x_2}, \dots, e^{x_s})\in\mathbb{R}^s$ and $\mathrm{diag}(x)$
denotes the diagonal matrix whose diagonal is the vector $x$.

As the system is consistent, there exists a positive vector $\lambda$ such that 
$\sum_{y \to y'}\lambda_{yy'} (y'-y)=0$.  For any edge $y \to y'$,
take the (positive) rate constant 
\[k_{y y'} = \lambda_{y y'} (\bx^1)^{-y},\]
which defines a positive vector $\Ka$ satisfying
\[f(\bx^1,\Ka)= \underset{y\to y'}{\sum} \kappa_{yy'} \,  (\bx^1)^y \, (y'-y) = 0.\]
Then, $\bx^1$ is a positive steady state of the system for these reaction
rate constants  $\Ka$. As the system is a toric MESSI system $\bx^1$ is a
solution of the binomial equations that describe the positive steady states.
Call $\boldsymbol{\eta}:=(x^1)^\binoM$. Then, $\bx$ is a positive steady state of
the system if and only if $\bx^\binoM=\boldsymbol{\eta}$.
It can be checked that $((x^2)^\binoM)_j=e^{\langle \vv,\binoM_j\rangle}(x^1)^{\binoM_j}$, and,
as $\vv\in\bino^\bot$, we have $(x^1)^\binoM=(x^2)^\binoM=\boldsymbol{\eta}$.
Therefore, $x^2$ is also a positive steady state
of the system. Moreover, $x^2-x^1=\ww\in\stoich$, and so
$x^1$ and $x^2$ belong to the same stoichiometric compatibility class.
\end{proof}

Recall the definitions of $S_1$ and $\tilde{S}_1$ before Remark~\ref{rem:n+p}.

\begin{lemma}\label{lem:S1}
Assume that condition $(\cond')$ in Definition~\ref{def:storic}  holds, and consider the vectors
 \begin{equation}\label{eq:vk}
  v_k=y_k-y_{i_kj_k}.
 \end{equation}
 Then, $\stoich=\tilde{\stoich}_1\oplus\langle v_1,\dots,v_p\rangle$.
\end{lemma}
\begin{proof}
It is clear, from the definitions of $\tilde{\stoich}_1$ and the vectors $v_k$,
that $\tilde{\stoich}_1\cap\langle v_1,\dots,v_p\rangle=\{0\}$ (as no intermediate
complex appears in the reactions of $G_1$).
Moreover, the vectors $v_k$ are linearly independent, and therefore
$\dim(\langle v_1,\dots,v_p\rangle)=p$.
By Remark~\ref{rem:n+p}, we know that $\dim(\stoich)=\dim(S_1)+p=\dim(\tilde{S}_1)+p$.
Thus, we only need to show now that $\stoich\supseteq \tilde{\stoich}_1\oplus\langle v_1,\dots,v_p\rangle$. 

For simplicity, we will assume that all core complexes consist of
two species, but it is easy to adapt the proof for the case where the core
complexes consist of only one species. 
We first notice that $v_k\in \stoich$ for all $k$.
In fact, if $X_{i_k}+X_{j_k}\uri U_k$, there exist
$U_{k_1},\dots,U_{k_t}$ intermediates such that the chain of reactions
$X_{i_k}+X_{j_k}\to U_{k_1}\to \dots \to U_{k_t}\to U_k$ is in $G$. Therefore, 
from the telescopic sum $y_k-y_{i_kj_k}=(y_k-y_{k_t})+(y_{k_t}-y_{k_t-1})+
\dots+(y_{k_2}-y_{k_1})+(y_{k_1}-y_{i_kj_k})$,
we see that $v_k \in \stoich$, as we wanted to prove.
Given $X_i+X_j\to X_\ell+X_m$ in $G_1$, there
exist intermediates $U_{k_1},\dots,U_{k_t}$ such that the chain of reactions
$X_i+X_j\to U_{k_1}\to \dots \to U_{k_t}\to X_\ell+X_m$ is in $G$. As above, from a telescopic sum
we deduce that $y_{\ell m}-y_{ij}\in \stoich$. Hence, $\tilde{\stoich}_1\subseteq\stoich$ and
$\tilde{\stoich}_1\oplus\langle v_1,\dots,v_p\rangle\subseteq\stoich$.
\end{proof}

\begin{proof}[Proof of Proposition~\ref{prop:rankM}]
By Theorem~\ref{th:par}, we know that $\rank(\stoichMp)=m$.
  We show now that $\rank(\binoM)=s-m$, or equivalently that $\mathrm{dim}(\bino)=s-m=p+n-m$.
 From~\eqref{eq:uk} we see that the vectors $v_k$ defined
 in~\eqref{eq:vk} live in $\bino$ for all $1\le k\le p$ (recall that $y_k$ denotes the vector
corresponding to the monomolecular complex $U_k$). This implies that
 $\langle v_1, \dots, v_p\rangle \subseteq \bino$.
 As none of the exponents determined by~\eqref{eq:b} involves any variable $u_i$,
 it is enough to find $n-m$ linearly independent vectors in $\bino$ that have support in
 the last $n$ coordinates.
 
 Call $\bino_x$ the projection $\pi_x(\bino)$ of $\bino$ onto the last $n$ coordinates
 corresponding to  $x_1,\dots,x_n$. We need to prove then that $\dim(\bino_x)=n-m$.
 For each $\alpha\ge 1$, fix $i_\alpha \in \Sp^{(\alpha)}$
 and for each $X_j\in \Sp^{(\alpha)}$, $j\neq i_\alpha$, call
 $z_{i_\alpha j}=(\gamma_j+e_{i_\alpha})-(\gamma_{i_\alpha}+e_j)$,
 the vector in $\R^n$ deduced from the exponents of the binomials in~\eqref{eq:b}.
Denote by $\bino_\alpha$ the linear subspace with generators $\{z_{i_\alpha j}\}_{j\neq i_\alpha}$. 
 We claim that $\dim(\bino_\alpha)=n_\alpha-1$ for any $\alpha\ge 1$ and that
$\bino_x=\bino_1\oplus \bino_2\oplus\dots \oplus \bino_m$ 

To prove these claims, we need to recall the proof of Theorem~\ref{th:bss}.
We consider again the subsets $L_0, L_1, \dots$, and we assume that $\alpha \in L_k$. Then, as remarked
in the last paragraph of that proof, it holds that the connected component $G_2^\alpha$ with
vertices in $\Sp^{(\alpha)}$ (ensured by Lemma~\ref{lem:min} by our hypothesis of minimality of the partition) 
has labels in $\Q[\tau, \bx_\beta:\beta\in L_t, t<k]$. This implies that the $j$th coordinate of
the vector $z_{i_\alpha h}$ equals $-1$ if $h=j$ and $0$ otherwise. So the vectors
$\{z_{i_\alpha j}\}_{j\neq i_\alpha}$ are linearly independent, that is, $\dim(\bino_\alpha)=n_\alpha-1$,
and by a similar argument we deduce that the sum is direct. 
Therefore, $\dim(\bino_x)=\sum_{\alpha=1}^m(n_\alpha-1)=n-m$, as wanted.
 
\end{proof}

\subsection*{Algorithm}\label{ap:algorithm}

Step~1 in the algorithm follows directly from Theorem~\ref{th:monostationarity}.
Step~7 follows from 
\cite{CoFlRa08,Fein95DefOne,PM12} and Theorem~\ref{th:multistationarity}. 
Theorem~\ref{th:toric_toric}
explains how to find a matrix $B$ for an s-toric MESSI system.
The intermediate steps follow from the following considerations.
Given a matrix $A$, every vector in $\mathrm{rowspan}(A)$ is a conformal sum of circuits. 
(We refer the reader to \cite{mure16,Rock69,St95}.)
Moreover, the circuits of a matrix $A\in \R^{d\times s}$ of rank $d$ are found in the following way. 
For $J\subseteq [s]$ with $\#J=d-1$,  define $r_J\in \mathrm{rowspan}(A)$ as the vector 
$r_{J,\ell}= (-1)^{\mu(\ell,J)} \det(A_{J \cup \{\ell\}})$, where
$\mu(\ell,J)$ is the sign of the permutation of $J \cup \{\ell\}$ which takes $\ell$ followed by
the ordered elements of $J$ to the ordered elements of $J \cup \{\ell\}$,  for all $\ell \in \{0,
\dots, s\}$.
The following lemma is straightforward and well known.

\begin{lemma}\label{lem:circuits}
Let $A\in \R^{d\times s}$ be a matrix of rank $d$ and $J\subseteq [s]$ such that $\#J=d-1$ and
$\mathrm{rank}(A_J)=d-1$. Then $r_J$ is a circuit of $A$. 
Moreover, up to a multiplicative constant, 
these are all the circuits of $A$ (possibly repeated).
\end{lemma}

 \section*{Acknowledgments}
 We would like to thank the referee for providing useful comments.


\end{document}